\documentclass[11pt]{article}

\usepackage[dvips]{color}
\usepackage[pdftex]{graphicx}
\usepackage{amsfonts}
\usepackage{amssymb}
\usepackage{amsmath}
\usepackage{amsthm}
\usepackage{latexsym}
\usepackage{enumerate}
\usepackage{hyperref}
\usepackage{color}
\usepackage{authblk}
\usepackage{array}
\usepackage{tablefootnote}

\newtheorem{theorem}{Theorem}[section]
\newtheorem*{theorem*}{Theorem}
\newtheorem{lemma}[theorem]{Lemma}
\newtheorem{conjecture}[theorem]{Conjecture}

\newtheorem{corollary}[theorem]{Corollary}
\newtheorem{definition}[theorem]{Definition}

\usepackage{mathtools}
\DeclarePairedDelimiter\ceil{\lceil}{\rceil}
\DeclarePairedDelimiter\floor{\lfloor}{\rfloor}

{\makeatletter
 \gdef\xxxmark{%
   \expandafter\ifx\csname @mpargs\endcsname\relax % in minipage?
     \expandafter\ifx\csname @captype\endcsname\relax % in figure/caption?
       \marginpar{xxx}% not in a caption or minipage, can use marginpar
     \else
       xxx % notice trailing space
     \fi
   \else
     xxx % notice trailing space
   \fi}
 \gdef\xxx{\@ifnextchar[\xxx@lab\xxx@nolab}
 \long\gdef\xxx@lab[#1]#2{\textbf{[\xxxmark #2 ---{\sc #1}]}}
 \long\gdef\xxx@nolab#1{\textbf{[\xxxmark #1]}}
}

\newcommand{\CACHE}{\textit{CACHE}}
\newcommand{\MEM}{\textit{MEM}}
\newcommand{\PCACHE}{\textit{PCACHE}}

\newcommand{\tO}{\tilde{O}}

\newcommand{\tOmg}{\tilde{\Omega}}

\let\epsilon=\varepsilon

\usepackage{fullpage}

\begin{document}
\title{Fine-grained I/O Complexity via Reductions:\\
	\emph{New lower bounds, faster algorithms, and a time hierarchy}}
\author[1]{Erik D. Demaine}
\author[1]{Andrea Lincoln} 
\author[1]{Quanquan C. Liu}  
\author[1]{Jayson Lynch}
\author[1]{\\Virginia Vassilevska Williams}
	\affil[1]{Computer Science and Artificial Intelligence Lab, Massachusetts Institute of Technology,
		  Cambridge, MA, USA \authorcr
          \url{{edemaine, andreali, quanquan, jaysonl, virgi}@mit.edu}}
\date{}
\maketitle

\setcounter{page}{1}
\section*{Abstract}
%In this paper we consider the popular theoretical model of caching, the I/O model (or external memory model), and ask ``What can fine-grained techniques and assumptions improve in the I/O model?''. Our results :
%Fine-grained complexity has generate many reductions between polynomial time problems. Fine-grained complexity has also generated many reduction techniques. These reductions have the potential to propagate lower bounds and faster algorithms between problems that have reductions from one to the other.

This paper initiates the study of I/O algorithms (minimizing cache misses)
from the perspective of fine-grained complexity
(conditional polynomial lower bounds).
Specifically, we aim to answer why sparse graph problems are so hard,
and why the Longest Common Subsequence problem gets a savings of a factor of
the size of cache times the length of a cache line, but no more.
%We seek to give explanations for running times, and improve running times using reductions. 
We take the reductions and techniques from complexity and fine-grained complexity and apply them to the I/O model to generate new (conditional) lower bounds
as well as faster algorithms.
% in the I/O model.
We also prove the existence of a time hierarchy for the I/O model, which motivates the fine-grained reductions.
%Along the way, we use similar kinds of reductions to actually improve the best known running times for a few problems, both in the I/O model and the RAM.

\begin{itemize}
	\item Using fine-grained reductions, we give an algorithm for distinguishing 2 vs.\ 3 diameter and radius that runs in $O(|E|^2/(MB))$ cache misses, which for sparse graphs improves over the previous $O(|V|^2/B)$ running time.
	\item We give new reductions from radius and diameter to Wiener index and median. These reductions are new in both the RAM and I/O models. 
	\item We show meaningful reductions between problems that have linear-time solutions in the RAM model. The reductions use low I/O complexity (typically $O(n/B)$), and thus help to finely capture the relationship between ``I/O linear time'' $\Theta(n/B)$ and RAM linear time $\Theta(n)$. 
	\item We generate new I/O assumptions based on the difficulty of improving sparse graph problem running times in the I/O model. We create conjectures that the current best known algorithms for Single Source Shortest Paths (SSSP), diameter, and radius are optimal.
	\item  From these I/O-model assumptions, we show that many of the known reductions in the word-RAM model can naturally extend to hold in the I/O model as well (e.g., a lower bound on the I/O complexity of Longest Common Subsequence that matches the best known running time). 
	\item We prove an analog of the Time Hierarchy Theorem in the I/O model, further motivating the study of fine-grained algorithmic differences.
\end{itemize}

%We define the classes $\PCACHE_{M,B}$ and $\CACHE_{M,B}\left(t\left(n\right)\right)$ describing the problems solvable in a polynomial number of cache misses and $O\left(t\left(n\right)\right)$ cache misses respectively. We demonstrate that $\PCACHE_{M,B}$ lies between $P$ and $PSPACE$ for reasonable choices of cache size. We provide simulations between the I/O model and both the RAM and Turing machine models. We prove the existence of a time hierarchy in $\CACHE_{M,B}\left(t\left(n\right)\right)$.
%

%We then define fine-grained reductions in both the IO-model and in the Hybrid models of caching. 
%We motivate this hierarchy and the reductions by showing relationships between 

\newpage

%% !TEX root =  ca-full.tex 

\section{Introduction}
\label{sec:intro}

The I/O model (or external-memory model) was introduced by Aggarwal and Vitter \cite{AggarwalVi88} to model the non-uniform access times of memory in modern processors. The model nicely captures the fact that, in many practical scenarios, cache misses between levels of the memory hierarchy (including disk) are the bottleneck for the program.
As a result, the I/O model has become a popular model for developing cache-efficient algorithms. 

%rEDUCTIONS ARE great
In the I/O model, the expensive operation is bringing a \emph{cache line} of $B$ contiguous words from the ``main memory'' (which may alternately represent disk) to the ``cache'' (local work space). The cache can store up to $M$ words in total, or $M/B$ cache lines. Computation on data in the cache is usually treated as free, and thus the main goal of I/O algorithms is to access memory with locality. That is, when bringing data into cache from main memory in contiguous chunks, we would like to take full advantage of the fetched cache line. This is preferable to, say, randomly accessing noncontiguous words in memory.

When taking good graph algorithms for the RAM model and analyzing them in the I/O model, the running times are often very bad. Take for example Dijkstra's algorithm or the standard BFS algorithm. These algorithms fundamentally look at an adjacency list, and follow pointers to every adjacent node. Once the new node is reached, the process repeats, accessing all the adjacent nodes in priority order that have not previously been visited. This behavior looks almost like random access! Unless one can efficiently predict the order these nodes will be reached, the nodes will likely be stored far apart in memory. Even worse, this optimal order could be very different depending on what node one starts the algorithm at. 

Because of this bad behavior, I/O-efficient algorithms for graph problems take a different approach. For dense graphs, one approach is to reduce the problems to matrix equivalent versions. For example, APSP is solved by $(\min,+)$ matrix multiplication~\cite{jia1981complexity,pagh2014inputMatrix,seidel92}. The locality of matrices leads to efficient algorithms for these problems. 

Unfortunately, sparse graph problems are not solved efficiently by $(\min,+)$ matrix multiplication. For example, the best algorithms for directed single-source shortest paths in sparse graphs take $O(n)$ time, giving no improvement from the cache line at all~\cite{Brodal2004,CR04,ABDHI07}. Even in the undirected case, the best algorithm takes $O(n/\sqrt{B})$ time in sparse graphs~\cite{Meyer2003}. 

The Diameter problem in particular has resisted improvement beyond $O(|V|/\sqrt{B})$ even in undirected unweighted graphs \cite{Mehlhorn2002}, and in directed graphs, the best known algorithms still run in time $\Omega(|V|)$ \cite{chiang1995external,ABDHI07}. For this reason, we use this as a conjecture and build a network of reductions around sparse diameter and other sparse graph problems. 

In this paper we seek to explain why these problems, and other problems in the I/O model, are so hard to improve and to get faster algorithms for some of them.

% io models whatever 
In this paper we use reductions to generate new algorithms, new lower bounds, and a time hierarchy in the I/O model. Specifically, we get new algorithms for computing the diameter and radius in sparse graphs, when the computed radii are small. We generate novel reductions (which work in both the RAM and I/O models) for the Wiener Index problem (a graph centrality measure). We generate further novel reductions which are meaningful in the I/O model related to sparse graph problems. Finally, we show that an I/O time hierarchy exists, similar to the classic Time Hierarchy Theorem.

\paragraph{Caching Model and Related Work.}
Cache behavior has been studied extensively. In 1988, Aggarwal and Vitter \cite{AggarwalVi88} developed the \emph{I/O model}, also known as the external-memory model \cite{Demaine02}, which now serves as a theoretical formalization for modern caching models. 
A significant amount of work on algorithms and data structures in this model has occurred including items like buffer-trees \cite{LarsArge2003}, B-trees \cite{BayerBTree70}, permutations and sorting \cite{AggarwalVi88}, ordered file maintenance \cite{Bender02}, $(\min,+)$ matrix multiplication \cite{jia1981complexity,pagh2014inputMatrix}, and triangle listing  \cite{pagh2014input}. 
Frigo, Leiserson, Prokop and Ramachandran \cite{FrigoLePr99} proposed the \emph{cache-oblivious model}. In this model, the algorithm is not given access to the cache size $M$, nor is it given access to the cache-line size $B$. Thus the algorithm must be oblivious to the cache, despite being judged on its cache performance. 
Some surveys of the work include \cite{Vitter01,Arge96,Demaine02}. 

%In the external-memory model, there is a cache of size $M$ words with cache lines of size $B$ and a large external memory. Every time a memory address is pulled into cache, $B$ contiguous words of $w$ bits will also be pulled into memory. Each of these memory block transfers takes one unit of time. RAM operations on the data in the cache are free (and thus take no time). We may use the terms cache misses and time interchangeably when discussing this model. %The cache can fit $m = M/B$ cache lines, each of which fits $B$ words, for a total of $M = mB$ words in cache. Capital $M$ is the number of words that fit in cache. Lower case $m$ is the number of cache lines in cache. All of $B$, $m$, and $w$ are required to be positive integers. %We can say that an algorithm takes a certain amount of time: for example it takes $O(N/B)$ time to read off the $N$ elements in order from a list and $O(\frac{N^3}{\sqrt{M}B})$ to do matrix multiplication on $N \times N$ matrices. %The model is expressed in Figure \ref{fig:idealcachemodel}. 

%\begin{figure}[ht]
%	\label{fig:idealcachemodel}
%	\centering
%	\includegraphics[width=0.5\textwidth]{EMMimage}
%	\caption{An image depicting the external-memory model from Demaine's paper \cite{Demaine02}.}
%\end{figure} 

%This paper uses the following definitions for $m,M,B,n,$ and $w$:
%
%\begin{itemize}
%\itemsep0em 
%\item Let $B \geq 1$ be the number of words that fit on a cache line,
%\item 	$M$ be the size of the cache measured in words,
%\item 	$n$ be the problem size, and
%\item 	$w \geq 1$ be the number of bits per word. 
%\end{itemize}

When requesting cache lines from main memory in this paper we will only request the $B$ words starting at location $xB$ for integers $x$. Another common model, which we do not follow in this paper, allows arbitrary offsets for the cache line pulls. This can be simulated with at most twice as many cache misses and twice as much cache. \looseness=-1

\paragraph{Fine-grained Complexity.}

A popular area of recent study is fine-grained complexity. The field uses efficient reductions to uncover relationships between problems whose classical algorithms have not been improved substantially in decades. Significant progress has been made in explaining the lack of progress on many important problems \cite{backurs2015edit, williams2010subcubic, radiusgrandoni, abboud2014popular, frechet, cfg, diametervrod} such as APSP, othogonal vectors (OV), 3-SUM, longest common subsequence (LCS), edit distance and more. Such results focused on finding reductions from these problems to other (perhaps less well-studied) problems such that an improvement in the upper bound on any of these problems will lead to an improvement in the running time of algorithms for these problems. For example, research around the All-Pairs Shortest Paths problem (APSP) has uncovered that many natural, seemingly simpler graph problems on $n$ node graphs are fine-grained equivalent to APSP, so that an $O(n^{3-\epsilon})$ time algorithm for $\epsilon>0$ for one of the problems implies an $O(n^{3-\epsilon'})$ time algorithm for some $\epsilon'>0$ for all of them.

\subsection{History of Upper Bounds}
\label{subsec:UBhist}

In the I/O model, the design of algorithms for graph problems is difficult. This is demonstrated by the number of algorithms designed for problems like Sparse All Pairs Shortest Paths, Breath First Search, Graph Radius and Graph Diameter where very minor improvements are made (see Table~\ref{table:problem-defs} for definitions of problems). Note that dense and sparse qualifiers in the front of problems indicate the problem defined over a dense/sparse graph, respectively. 

The Wiener index problem measures the total distance from all points to each other. Intuitively, this measures how close or far points in the graph are from each other. In this respect, Wiener index is similar to the radius, diameter and median measures of graph distance. 

\begin{table}
\begin{center}
\begin{tabular}{ |c|c| } 
\hline
Problem Name & Problem Definition\\
\hline
\hline
Orthogonal & Given two sets $U$ and $V$ of $n$ vectors each with\\
Vector (OV) & elements $\{0, 1\}^d$ where $d = \omega(\log{n})$, determine whether there\\
& exist vectors $u \in U$ and $v \in V$ such that $\sum_{i = 1}^d u_i \cdot v_i = 0$.\\
\hline
Longest Common & Given two strings of $n$ symbols over some alphabet $\Sigma$, compute the length \\
Subsequence (LCS) &  of the longest sequence that appears as a subsequence in both input strings.\\
\hline
Edit Distance (ED) & Given two strings $s_1$ and $s_2$, determine the minimum\\
& number of operations that converts $s_1$ to $s_2$. \\
\hline
Sparse Diameter & Given a sparse graph $G = (V, E)$, determine if $\max_{u, v \in V}d(u, v)$ \\
& where $d(u, v)$ is the distance between nodes $u$ and $v$ in $V$.\\
\hline
2 vs. 3 Sparse & Given a sparse graph $G = (V, E)$, determine if $\max_{u, v \in V}d(u, v) \leq 2$ \\
 Diameter & where $d(u, v)$ is the distance between nodes $u$ and $v$ in $V$.\\
\hline
Hitting Set (HS) & Given two lists of size $n$, $V$ and $W$, where the elements are taken\\
& from a universe $U$, does there exist a set in $V$ that hits (contains an\\ 
& element of) every set in W. \\
\hline
Sparse Radius & Given a sparse graph $G = (V, E)$, determine $\min_{u\in V}\left(\max_{v \in V} d(u, v)\right)$\\
&  where $d(u, v)$ is the distance between nodes $u$ and $v$ in $V$.\\
\hline
2 vs. 3 Sparse  & Given a sparse graph $G = (V, E)$, determine if $\min_{u\in V}\left(\max_{v \in V} d(u, v)\right) \leq 2$\\
Radius &  where $d(u, v)$ is the distance between nodes $u$ and $v$ in $V$.\\
\hline
3 vs. 4 Sparse & Given a sparse graph $G = (V, E)$, determine if $\min_{u\in V}\left(\max_{v \in V} d(u, v)\right) \leq 3$\\
Radius &  where $d(u, v)$ is the distance between nodes $u$ and $v$ in $V$.\\
\hline
Median & Let $d(u,v)$ be the shortest path distance between nodes $u$ and $v$ in a graph $G$.\\
& The median is the node $v$ that minimizes the sum $\sum_{u\in V} d(v,u)$.\\
\hline
3-SUM & Given a set of $n$ integers, determine whether the set \\
& contains three integers $a, b, c$ such that $a + b = c$.\\
\hline
Convolutional  & Given three lists $A$, $B$ and $C$ each consisting of $n$ numbers,\\ 
3-SUM & return true if $\exists i,j,k \in [0,n-1]$ such that \\
& $i+j+k \equiv 0 \textbf{ mod } n$ and $A[i]+B[j]+C[k] = 0$.\\
\hline
0 Triangle & Given a graph $G$, return true if $\exists a,b,c \in V$ such that\\
& $w(a,b)+w(b,c)+w(c,a)=0$ where $w(u,v)$ is the weight\\
& of the edge $(u,v)$.\\
\hline
All-Pairs Shortest & Given a directed or undirected graph with integer weights, \\ 
Paths (APSP) & determine the shortest distance between all pairs of \\
&vertices in the graph.\\
\hline
Wiener Index & Let $d(u,v)$ be the shortest paths distance between nodes $u$\\ & and $v$ in a graph $G$. The Wiener Index of $G$ is $\sum_{u\in V}\sum_{v\in V} d(u,v)$.\\
\hline
Negative Traingle & Given a graph $G$, return true if $\exists a,b,c \in V$ such that\\
&  $w(a,b)+w(b,c)+w(c,a)<0$ where $w(u,v)$ is the weight of the edge $(u,v)$.\\
\hline 
$(\min, +)$-Matrix  & Given two $n$ by $n$ matrices $A$ and $B$, \\
Multiplication &  return $C[i,k] = \min_{j\in[1,n]} \left(A[i,j]+B[j,k]\right)$.\\
\hline
Sparse Weighted  & Given a sparse, weighted graph $G = (V, E)$, determine \\
Diameter & $\max_{u, v \in V}d(u, v)$ where $d(u, v)$ is the distance\\
& between nodes $u$ and $v$ in $V$.\\
\hline
\end{tabular}
\caption{Fine-grained problems definitions.}\label{table:problem-defs}
\end{center}
\end{table}

The history of improvements to the upper bound of negative triangle in the I/O model is an important example of the difficulty in the design of I/O efficient algorithms for graph problems (see Table~\ref{table:apsp-results} for a summary). For a long time, no improvements in terms of $M$ were made to the upper bound for negative triangle. A key was re-interpreting the problem as a repeated scan of lists.

\begin{table}
	\[
	\renewcommand\arraystretch{1.5}
	\begin{array}{|c|c||c|c||c|c|} 
	\hline
	\multicolumn{2}{|c||}{\text{Dense APSP}} & \multicolumn{2}{c||}{\text{Sparse Weighted Diameter}} & \multicolumn{2}{c|}{\text{Sparse}-\triangle} \\
	\hline
	\tO(n^3) & \text{[naive]}  & \tO(n^2) & \text{[naive]}& \tO(n^{1.5})& \text{[naive]}\\
	\tO(n^3/(\sqrt{M})) & \text{\cite{HongKu81}} & \tO(n^2/\sqrt{B}) & \text{\cite{apspSparse}} & \tO(n+n^{1.5}/B) & \text{\cite{menegola2010external}}\\
	\tO(n^3/(\sqrt{M}B)) & \text{Extension}^* & \tO(n^2/\sqrt{B}) & \text{\cite{chowdhury2005external}}  & \tO(n^{1.5}/B) & \text{\cite{dementiev2006algorithm}}\\
	&  & & & \tO(n^{1.5}/(\sqrt{M}B)) & \text{\cite{pagh2014input}}\\
	\hline
	\end{array}
	\]
	\caption{History of APSP upper bounds. $^*$ This is an extension of \cite{HongKu81} see e.g. \cite{pagh2014input}.}\label{table:apsp-results}
\end{table}

We feel that the history of negative triangle has taught us that upper bounds in the I/O model on graph problems are best achieved by creating efficient reductions from a graph problem to a non-graph problem. Hence the study of fine-grained reductions in the I/O model is crucial to using this approach in solving such graph problems with better I/O efficiency. Graph problems tempt the algorithmic designer into memory access patterns that look like random access, whereas matrix and array problems immediately suggest memory local approaches to these problems. We consider the history of the negative triangle problem to be an instructive parable of why the matrix and array variants are the right way to view I/O problems.

\subsection{Our Results}
We will now discuss our results in this paper. In all tables in this paper, our results will be in bold. We demonstrate the value of reductions as a tool to further progress in the I/O model through our results. 

Our results include improved upper bounds, a new technique for lower bounds in the I/O model and the proof of a computational hierarchy. Notably, in this paper we tie the I/O model into both fine-grained complexity and classical complexity. 
\subsubsection{Upper Bounds}

We get improved upper bounds on two sparse graph problems and have a clarifying note to the community about matrix multiplication algorithms. 

For both the sparse 2 vs 3 diameter and sparse 2 vs 3 radius problems, we improve the running time from $O(n^2/B)$ to $O(n^2/(MB))$. We get these results by using an insight from a pre-existing reduction to two very local problems which have trivial $O(n^2/(MB))$ algorithms that solve them. Note that this follows the pattern we note in Section \ref{subsec:UBhist} in that we produce a reduction from a graph problem to a non-graph problem to obtain better upper bounds in terms of $M$. 

Furthermore, previous work in the I/O model related to matrix multiplication seems to use the naive matrix multiplication $n^3$ bound, or the Strassen subdivision. However, fast matrix multiplication algorithms which runs in $n^{\omega}$ time imply a nice self-reduction.  Thus, we can get better I/O algorithms which run in the most recent fast matrix multiplication time. We want to explicitly add a note in the literature that fast matrix multiplication in the I/O model should run in time $T_{MM}(n,M,B) =O( n^{\omega'}/(M^{\omega'/2-1}B))$ where $\omega'$ is the matrix multiplication exponent, if it is derived using techniques bounding the rank of the matrix multiplication tensor. The current best $\omega'$ is $\omega'<2.373$~\cite{vstoc12,legall} giving us the I/O running time of $T_{MM}(n,M,B) = O(n^{2.373}/(M^{0.187}B))$. We give these results in Section~\ref{sec:MMspeed}.

\subsubsection{I/O model Conjectures}

In the I/O model a common way to get upper bounds is to get a self-reduction where a large problem is solvable by a few copies of a smaller problem. We make the small subproblems so small they fit in cache. If the problem is laid out in a memory local fashion in main memory then it will take $M/B$ I/Os to solve a subproblem that fits in memory $M$. %If the space usage for a problem of size $n$ is $S(n)$ then we want to make the subproblems of size $S^{-1}(M)$. We will go over two examples. 

In Section \ref{sec:masterstheorem}, we give an I/O-based Master Theorem which gives the running time for algorithms with recurrences of the form $T(n,M,B) = \alpha T(n/\beta,M,B)  + f(n,M,B)$ (like the classic Master Theorem from \cite{CLRS09}) and $T(n, M, B) = g^2T(n/\beta, M, B) + f(n, M, B)$ (self-reduction). The running times generated by these recurrences match the best known running times of All-Pairs Shortest Paths (APSP), 3-SUM, Longest Common Subsequence (LCS), Edit Distance, Orthogonal Vectors (OV), and more. Thus, if we conjecture that a recursive algorithm has a running time that is optimal for a problem, we are able to transfer this bound over to the I/O model using our Master Theorem and self-reduction framework in a natural way. %The results from standard recursive algorithms using our Master Theorem and self-reduction framework are shown in Table~\ref{table:mt-results} and the proofs can be found in Section~\ref{sec:mt-proofs}. 

\paragraph{Lower Bounds From Fine-Grained Complexity Assumptions.}

We demonstrate that many of the reductions in the RAM model between problems of interest and common fine-grained assumptions give lower bounds in the I/O model. We generate reasonable I/O conjectures for these problems and demonstrate that the reductions are I/O-efficient. First, we begin with the conjectures. 

\begin{conjecture}[I/O All-Pairs Shortest Paths (APSP) Conjecture]\label{conj:apsp}
 APSP requires $\frac{n^{3-o(1)}}{M^{1/2+o(1)}B^{1+o(1)}}$ I/Os. 
\end{conjecture}

\begin{conjecture}[I/O $3$-SUM Conjecture]\label{conj:3-sum}
	3-SUM requires $\frac{n^{2-o(1)}}{M^{1+o(1)}B^{1+o(1)}}$ I/Os. 
\end{conjecture}

\begin{conjecture}[I/O Orthogonal Vectors (OV) Conjecture]\label{conj:ov}
	OV requires $\frac{n^{2-o(1)}}{M^{1+o(1)}B^{1+o(1)}}$ I/Os. 
\end{conjecture}

\begin{conjecture}[I/O Hitting Set (HS) Conjecture]\label{conj:hs}
	HS requires $\frac{n^{2-o(1)}}{M^{1+o(1)}B^{1+o(1)}}$ I/Os. 
\end{conjecture}

From these conjectures we can generate many lower bounds. Many of our lower bounds are tight to the fastest known algorithms. These reductions have value even if the conjectures are refuted since many of these reductions also give upper bounds for other problems--leading to better algorithms for many problems even if the conjectures are refuted.\\
\\

\begin{table}
\begin{minipage}{\linewidth}
\[
\renewcommand\arraystretch{1.5}
\arraycolsep=3pt
\begin{array}{|c|c|c|c|c|c|} 
\hline
\text{Problem} & \text{Upper Bound} & \text{UB source} &\text{Lower Bound} &  \text{LB from} & \text{LB source} \\
	\hline
	\text{OV}&\tO(n^2/(MB))&\textbf{Lem \ref{lem:ovUB}}&\tOmg(n^2/(MB))&\text{I/O OV Conj}&\text{By Def}\\
	\text{LCS}&\tO(n^2/(MB))&\cite{chowdhury2006cache} &\tOmg(n^2/(MB))& \text{I/O OV Conj}&\textbf{Lem \ref{lem:LCSOV}}\\
	\text{Edit Distance}&\tO(n^2/(MB))&\cite{ChowdhuryRa06}&\tOmg(n^2/(MB))& \text{I/O OV Conj}&\textbf{Lem \ref{lem:edOV}}\\
	\text{Sparse Diameter}&\tO(n^2/B)& \cite{apspSparse} &\tOmg(n^2/(MB))& \text{I/O OV Conj}&\textbf{Lem \ref{lem:diamOV}}\\
	\text{2 vs.\ 3  Sprs.\ Diameter}&\tO(n^2/(MB))&\textbf{Lem~\ref{thm:diamCache}} &\tOmg(n^2/(MB))& \text{I/O OV Conj}&\textbf{Lem \ref{lem:diamOV}}\\
	\text{Hitting Set}&\tO(n^2/(MB))&\textbf{Lem \ref{lem:hsUB}} &\tOmg(n^2/(MB))& \text{I/O HS Conj}&\text{By Def.}\\
	\text{Sparse Radius}&\tO(n^2/B)& \cite{apspSparse} &\tOmg(n^2/(MB))& \text{I/O HS Conj}&\textbf{Lem \ref{lem:radiHS}}\\
	\text{2 vs.\ 3 Sparse Radius}&\tO(n^2/(MB))&\textbf{Thm~\ref{thm:radiusCache}} &\tOmg(n^2/(MB))& \text{I/O HS Conj}&\textbf{Lem \ref{lem:radiHS}}\\
	\text{3 vs.\ 4 Sparse Radius}&\tO(n^2/B)
&\cite{apspSparse} &\tOmg(n^2/(MB))& \text{I/O HS Conj}&\textbf{Lem \ref{lem:radiHS}}\\
	\text{Sparse Median}&\tO(n^2/B)\footnote{Upper bound comes from APSP in general graphs.}
&\cite{apspSparse} &\tOmg(n^2/(MB))& \text{I/O HS Conj}&\textbf{Thm \ref{thm:3v4radius-to-median}}\\
\text{Sparse Median}&\tO(n^2/B)\footnote{Upper bound comes from APSP in general graphs.}
&\cite{apspSparse} &\tOmg(n^2/B)& \text{3 vs.\ 4 Sprs.\ Radius}&\textbf{Thm \ref{thm:3v4radius-to-median}}\\
	\text{3-SUM}&\tO(n^2/(MB))&\cite{baran2005subquadratic}&\tOmg(n^2/(MB))&\text{I/O 3-SUM Conj}&\text{By Def}\\
	\text{Conv. 3-SUM}&\tO(n^2/(MB))\footnote{The upper bound is a trivial extension of the 3-SUM upper bound for explanation see Lemma \ref{lem:0convReduction}.}&\cite{baran2005subquadratic} 
   &\tOmg(n^2/(MB))&\text{I/O 3-SUM Conj}&\textbf{Lem \ref{lem:con3sumred}}\\
	\text{0 Triangle}&\tO(n^3/(\sqrt{M}B))&\textbf{Lem \ref{lem:0triUB}}&\tOmg(n^2/(MB))&\text{I/O 3-SUM Conj}&\textbf{Thm \ref{thm:0tri3sum}}\\
	\text{APSP}&\tO(n^3/(\sqrt{M}B))&\cite{pagh2014inputMatrix} &\tOmg(n^3/(\sqrt{M}B))&\text{I/O APSP Conj}&\text{By Def}\\
	\text{Wiener Index}&\tO(n^3/(\sqrt{M}B))\footnote{This upper bound comes directly from applying the algorithm for APSP, solving APSP, and summing the results.} &\cite{pagh2014inputMatrix} &\tOmg(n^3/(\sqrt{M}B))&\text{I/O APSP Conj}&\textbf{Thm~\ref{thm:apsp-lb-wi}}\\
	\text{0 Triangle}&\tO(n^3/(\sqrt{M}B))&\textbf{Lem \ref{lem:0triUB}}&\tOmg(n^3/(\sqrt{M}B))&\text{I/O APSP Conj}&\textbf{Thm \ref{lem:zeroTraingleAPSP}}\\
	\text{$-$ Triangle}&\tO(n^3/(\sqrt{M}B))&\textbf{Thm \ref{thm:fastAlgorithms}}&\tOmg(n^3/(\sqrt{M}B))&\text{I/O APSP Conj}&\textbf{Thm \ref{thm:minTOapsp}}\\
	\text{(min,+) MM}&\tO(n^3/(\sqrt{M}B))&\cite{pagh2014inputMatrix}&\tOmg(n^3/(\sqrt{M}B))&\text{I/O APSP Conj}&\textbf{Thm \ref{thm:minTOapsp}} \\
%	\text{SSSP}&\tO(??)&\cite{??}&\tOmg(??))&\text{?? Conj}&\textbf{Thm \ref{??}} \\
%	\text{dir unweighted SSSP}&\tO(??)&\cite{??}&\tOmg(??))&\text{?? Conj}&\textbf{Thm \ref{??}} \\
%	\text{unweighted SSSP}&\tO(n/\sqrt{B})&\cite{??}&\tOmg(??))&\text{?? Conj}&\textbf{Thm \ref{??}} \\
%	\text{sparse APSP}&\tO(??)&\cite{??}&\tOmg(??))&\text{?? Conj}&\textbf{Thm \ref{??}} \\
%	\text{Connectivity}&\tO(??)&\cite{??}&\tOmg(??))&\text{?? Conj}&\textbf{Thm \ref{??}} \\
%	\text{s-t Shortest Path}&\tO(??)&\cite{??}&\tOmg(??))&\text{?? Conj}&\textbf{Thm \ref{??}} \\
%	\text{directed s-t Shortest Path}&\tO(??)&\cite{??}&\tOmg(??))&\text{?? Conj}&\textbf{Thm \ref{??}} \\
%	\text{Girth Through Edge}&\tO(??)&\cite{??}&\tOmg(??))&\text{?? Conj}&\textbf{Thm \ref{??}} \\
%	\text{Girth Through Vertex}&\tO(??)&\cite{??}&\tOmg(??))&\text{?? Conj}&\textbf{Thm \ref{??}} \\
%	\text{directed Girth Through Edge}&\tO(??)&\cite{??}&\tOmg(??))&\text{?? Conj}&\textbf{Thm \ref{??}} \\
%	\text{directed Girth Through Vertex}&\tO(??)&\cite{??}&\tOmg(??))&\text{?? Conj}&\textbf{Thm \ref{??}} \\
	\hline
\end{array}
\]
\end{minipage}
\caption{Previous results and our results on upper and lower bounds of problems. Sparse (Sprs.)\ means that $|E| = O(|V|)$.}\label{table:long-list}
\end{table}

\subsubsection{Lower Bounds from Sparse Graph Problems}
%\xxx{TODO: write up our linear time bounds and our diameter bounds}

In addition to the upper, lower bounds, and reductions presented in the I/O model for the standard RAM problems listed in Table~\ref{table:long-list}, we introduce novel upper, lower bounds, and reductions between graph problems. The reason for this focus is the fact that, more than in the RAM model, the I/O model has a history of particularly slow algorithms in graphs. In particular, sparse graph problems have very slow algorithms. We make novel reductions between sparse graph problems, many of which apply to the RAM model as well, such that solving one of these problems will solve many other variations of hard sparse graph problems in the I/O model. 

We provide reductions between problems that currently require $\Omega(n/\sqrt{B})$ time to solve. Thus, these problems specifically require linear time reductions. We show equivalence between the following set of problems for undirected/directed and unweighted graphs: $(s, t)$-shortest path, finding the girth through an edge, and finding the girth through a vertex.

We additionally generate a new reduction from sparse weighted Diameter to the sparse Wiener Index problem in Section \ref{sec:newLB}. This reduction holds in the RAM model as well as the I/O model.

%\paragraph{Reductions between $O(n^2/B)$ problems.}

%\xxx{TODO: make a figure of our $n^2/B$ }
%\xxx{TODO: make a table of our $n^2/B$ hard problems }

%\paragraph{Reductions between $O(n/\sqrt{B})$ problems.}

%\begin{figure}[ht]
%	\centering
%	\includegraphics[width=1\textwidth]{figures/reductions_diagram.png}
%	\caption{If A reduces to B we have an arrow from B to A. }
%	\label{fig:linearReductions}
%\end{figure}

%While, for example, both st shortest paths 

%\input{table-connectivity}

\subsubsection{Hierarchy}
The time and space hierarchy theorems are fundamental results in computational complexity that tell us there are problems which can be solved on a deterministic Turing Machine with some bounded time or space, which cannot be solved on a deterministic Turing Machine which has access to less time or space. See, notably, the famous time and space hierarchies \cite{sipser2006introduction}. For some classes, for example BPP, no time hierarchy is known to exist (e.g., \cite{ryanCommentBPP,Barak02}). 

In Section~\ref{sec:Hierarchy}, we show similar separation hierarchies exist in the I/O model once again using the simulations between the RAM and I/O models and our complexity class $\CACHE_{M,B}(t\left(n\right))$ defined in Section~\ref{sec:PandPSPACE} as the set of problems solvable in $O\left(t\left(n\right)\right)$ cache misses.

\begin{theorem}
	If the memory used by the algorithm is referenceable by $O(B)$ words (i.e. the entire input can be brought into cache by bringing in at most $O(B)$ words), then
	$$\CACHE_{M,B}(t\left(n\right))  \subsetneq CACHE_{M,B}(\left( t\left(n\right)B \right)^{1+\epsilon}).$$
\end{theorem}

Notably, this theorem applies any time we use a polynomially size memory and our word size is $w = \Omega(\lg n)$, which is the standard case in the RAM model. 

 This separation is motivation for looking at complexity of specific problems and trying to understand what computational resources are necessary to solve them. %Appendix~\ref{sec:TMsimRAM} explores some of these reductions further and shows where results for classic problems in complexity theory would lead to improved I/O algorithms and a better understanding of memory locality.

\subsubsection{Improved TM Simulations of RAM Imply Better Algorithms}
In Section~\ref{sec:TMsimRAM}, we show that improved simulations of RAM machines by Turing Machines would imply better algorithms in the I/O model.
Specifically, if we can simulate RAM more efficiently with either multi-tape Turing machines or multi-dimensional Turing machines, then we can show that we can gain some cache locality and thus save by some factor of $B$, the cache line size.

\subsection{Organization}
In this paper, we argue that the lens of reductions offer a powerful way to view the I/O model. We show that reductions give novel upper and lower bounds. We also define complexity classes for the I/O model and prove a hierarchy theorem further motivating the analysis of the I/O model using fine-grained complexity.

We begin with faster algorithms obtained through reductions which are collected in Section~\ref{sec:Algorithms}.
Section \ref{sec:newUpper} develops such algorithms for small
diameter and radius.
Section \ref{sec:masterstheorem} develops the I/O Master Theorem, which is more broadly a useful tool for analyzing almost all cache-oblivious algorithms.
Section~\ref{sec:MMspeed} uses this theorem to show how all recent improvements to matrix multiplication's RAM running time also give efficient cache-oblivious algorithms.

One can get new lower bounds, by using the techniques from fine-grained complexity. Some fine-grained reductions from the RAM model also work in the I/O model, we show examples in Sections \ref{sec:3sum}, \ref{sec:APSPReductions}, \ref{sec:ov}. We also get new reductions that work in both the RAM and I/O model related to the Wiener Index problem in Section \ref{sec:newLB}. Some reductions in the RAM model do not work in the I/O model; thus, in Section \ref{sec:newLB}, we give novel reductions between several algorithms which take $O(n^2/B)$ and $O(n^2/(MB))$ time. We also get reductions that are meaningful in the I/O model which are not in the RAM model, notably, between problems whose fastest algorithms are $O(n/\sqrt{B})$ and $O(n)$, respectively, in Section \ref{sec:linear}.

One can also use reductions and simulation arguments to prove a hierarchy theorem for the I/O model, explained in Section~\ref{sec:Pcache}.

\section{Algorithms in the I/O Model}
\label{sec:Algorithms}
In this section, we discuss our improved algorithms, algorithm analysis tools, and how reductions generate algorithms. As is typical in the I/O model, we assume that all inputs are stored in disk and any computation done on the inputs are done in cache (after some or all of the inputs are brought into cache). Section~\ref{sec:newUpper} gives better algorithms for the 2 vs 3 Diameter problem and the 2 vs 3 Radius problem in the I/O model. Section~\ref{sec:MMspeed} gives improved algorithms for Matrix Multiplication in the I/O model. 

Self-reductions are commonly used for cache-oblivious algorithms, because dividing until the subproblems are arbitrarily small allows for the problems to always fit in cache. In the RAM model, self-reductions allow for easy analysis via the Master Theorem. Despite the amount of attention to analyzing self-reductions in the I/O model, no one has written down the I/O-based Master Theorem. In Section~\ref{sec:masterstheorem}, we describe and prove a version of the Master Theorem for the I/O model. We present a proof of this theorem to simplify our analysis and to help future papers avoid redoing this analysis. 

Finally, in Section~\ref{sec:Hy1} we explain how some reductions in the RAM model imply faster algorithms in the I/O model. 
\subsection{Algorithms for Sparse 2 vs 3 Radius and Diameter}
\label{sec:newUpper}
For both the radius and diameter problems on unweighted and undirected graphs, we can show distinguishing between a diameter or radius of 2 and a larger diameter or radius can be solved efficiently.  Our algorithm relies on the reinterpretation of the 2 vs 3 problem as a set-disjointness problem. Every node, $v$, has an associated set, $S_v$, its adjacency list union itself. If two nodes have disjoint sets $S_v$ and $S_u$, then they are distance greater than 2 from each other. 
Our algorithm for 2 vs 3 diameter and radius save an entire factor of $M$ from the previously best known running times. 

This is a similar idea to the reduction from 2 vs 3 diameter to OV and from 2 vs 3 radius to Hitting Set in the RAM model. These reductions were introduced by Abboud, Vassilevska-Williams and Wang~\cite{diameterReduction}. While these reductions exist in the RAM model, they don't result in faster algorithms for 2 vs 3 diameter and radius in the I/O model because they use a hashing step that results in BFS being run from $\frac{|E|}{\Delta}$ nodes for some parameter $\Delta$ that can be set that gives the orthogonal vectors instance a dimension of $\Delta^2$. In the I/O model, BFS is quite inefficient: we would need to set $\Delta \approxeq M $ to get an efficient algorithm using the approach in~\cite{diameterReduction}. But, with a dimension of $M^2$ the algorithm will run very slowly. Therefore, below we present a solution to the set disjointness problem with no hashing into a smaller dimension.

Below we present the cache-aware algorithm for distinguishing 2 vs 3 diameter in an undirected, unweighted graph which runs in $O\left(\frac{|E|^2}{MB}+ sort(|E|)\right)$ time where $sort(|E|)$ is the time to sort the elements in $|E|$ in the I/O model and $sort(|E|) = O\left(\frac{|E|\log{|E|}}{B}\right)$. We leave the proofs of cache-oblivious 2 vs 3 diameter and radius in Appendix \ref{sec:newUpperProofs}. For both 2 vs 3 diameter and radius we get running times of $O\left(\frac{n^2}{MB}+ \frac{|E|\lg(|E|)}{B}\right)$. The algorithms and proofs for cache-obliviousness are finicky, but fundamentally are self-reductions of the form $T(n) = 4T(n/2)+n/B$. We leave the proofs to the Appendix because uneven subdivision and tracking bits are not very illuminating to the overall scope of our paper. 

%\section{Proofs for 2 vs 3 Diameter and Radius}
%\label{sec:DiamProofs}
%We will now give the algorithm for faster diameter and radius. 

%\subsection{2 vs 3 Diameter}

We will start by giving a non-oblivious algorithm which relies on a recursive self-reduction. We will then show how to make this oblivious. It is easier to explain the analysis and algorithm when we can rely on the size of cache, but we can avoid that and get an oblivious algorithm anyway. The previous best algorithm is from Arge, Meyer and Toma which achieves $O(|V|sort(|E|)) = O(\frac{|V| |E| \log{|E|}}{B})$~\cite{apspSparse}. We get an improvement over the previous algorithm in running time whenever $\frac{|E|}{|V|} = o(M)$.

\begin{theorem}\label{lem:2v3-sparse}
	Determining if the diameter of an undirected, unweighted graph is 1, 2 or greater than 2 can be done in $O\left(\frac{|E|^2}{MB}+ sort(|E|)\right)$ time in the I/O-model. 
\end{theorem}
\begin{proof}
	Check if the number of edges is $n^2$, this can be done in $O\left(\frac{|E|}{B}\right)$ time. If it is equal to $n^2$, then return $1$. 
	
	Otherwise, 
%sort the adjacency lists of all nodes in time $sort(E) = O\left(\frac{|E|\log(|E|)}{B}\right)$.
	count the size of the adjacency lists of all nodes and record these sizes in disk in time $O\left(\frac{|E|}{B}\right)$. 
	Sort the nodes in each adjacency list by the length of each of their adjacency lists in $O(sort(|E|)) = O(\frac{|E|\log{|E|}}{B})$ time. 
	Give each node an extra indicator bit $alreadyClose$. 
	
	We split the nodes into those with adjacency lists of length less than or equal to $M/4$ and those with adjacency lists longer than $M/4$. We call these the \emph{short} and \emph{long} adjacency lists, respectively. Let $A_S$ be the ordered set created by concatenating short lists ordered by length from shortest to longest and $A_L$ be the ordered set created by concatenating long lists also ordered by length.

	Sub-divide $A_S$ into subsections of length at least $M/4$ and less than $M/2$ in the following way. The $i$-th (for $i \in \left[0, \frac{8|E|}{M}\right]$) subsection contains some or all of the nodes whose adjacency lists are represented in the range of nodes from index $\frac{iM}{2}$ to $\frac{(i+2)M}{2}$ or some subset of the nodes in $A_S\big[\frac{iM}{2}: \frac{(i+2)M}{2}\big]$ (recall that $A_S$ and $A_L$ are lists of \emph{concatenated adjacency lists}--thus, any index into either list returns a node). 
	%and includes a subset of the adjacency lists that begin and end in this range. 
If an adjacency list begins at $A_S\big[\frac{j_1M}{2}\big]$ where $i -1 \leq j_1 < i$ and ends at $A_S\big[\frac{l_1M}{2}\big]$ where $i-1 \leq l_1 < i + 1$ ignore that adjacency list and do not include it in subsection $i$. If an adjacency list begins between  $A_S\big[\frac{j_2M}{2}\big]$ where $i \leq j_2 < i+1$ and $A_S\big[\frac{l_2M}{2}\big]$ where $i \leq l_2 < i+2$ then include it in subsection $i$. A given subsection can have length at most $\frac{M}{4}+\frac{M}{4} \leq \frac{M}{2}$.
	Then, create a copy of $A_S$ called $C_S$ where $C_S[i]$ contains the subsection of index $i$. We create a copy for convenience since we want to maintain the original $A_S$ while modifying $C_S$ and copying is cheap here. Furthermore, $C_S$ is different from $A_S$ in the sense that $C_S$ is an array of arrays (since it maintains the subsections we created from $A_S$ using the procedure above).
	%Furthermore, this is convenient notationally and for ease of explanation since we sometimes compare elements in $A_S$ with elements in $C_S$. 
	However, if one wants to be more efficient, one can perform the rest of the algorithm more carefully and can directly modify $A_S$ instead of $C_S$.
Also note there are at most $\frac{(2|E|)4}{M} = \frac{8|E|}{M}$ subsections in $C_S$.

	Suppose there are $k$ long adjacency lists in $A_L$. Let $A^i_L$ denote the section of $A_L$ representing the $i$-th adjacency list where $i\in [1,k]$. Then, if $A^i_L$ is longer than $M/4$, subdivide it into pieces $A^i_L[j]$ where $j \in \Big[1, \Big\lceil\frac{|A_L[i]|4}{M} \Big\rceil\Big]$.
	Create a copy of  $A_L$ called $C_L$ where $C_L[i]$ contains the subsection of $A_L$ with index $i$ (again this copy is an array of arrays). Note there are at most $\frac{8|E|}{M}$ subdivisions in total in $C_L$. 
	
	We would like to check if any of the nodes with long adjacency lists are far from other nodes. Ideally we would just run BFS from each node, but BFS in sparse graphs runs slowly (by a multiplicative factor of $\sqrt{B}$) in the I/O model. So, we will instead use a method of scanning through these lists. 
	
	First we will check if any two long lists are far from each other.  
	For $i\in[1,k]$ we set  $v_i$ to be the node associated with adjacency list $C_L[i]$. For every $j\in [1,k]$ we scan through the  adjacency lists $C_L[i]$ and $C_L[j]$ in sorted order progressing simultaneously in both lists to see if the intersection of the sets of nodes (each list also includes $v_i$ and $v_j$, respectively) in these two adjacency lists is non-empty. If all are close (intersection non-empty), we move on. If any are far (intersection empty), we return that the diameter is $>2$. This takes time $\sum_{i\in[1,k]}\sum_{j\in[1,k]}\frac{\left(|C_L[i]|+|C_L[j]|\right)}{B} =\frac{k|E|}{B}$ and $k\leq \min\left(\frac{8|E|}{M}, |V|\right)$. Therefore, the overall running time of this procedure is $O\left(\min\left(\frac{|E|^2}{MB}, \frac{|V||E|}{B}\right)\right)$.
	
	Next, we check if any long lists are far from short lists. For every $i\in[1,k]$ we will check if the associated node with the long list, call it $v_i$, is far from short nodes. For every $j\in \big[1, \frac{8|E|}{M}\big]$, we set each $v_j$'s $alreadyClose=False$ for every node, $v_j$ for $j \in \left[1, \frac{8|E|}{M}\right]$. For every subsection of size $M/4$ within $C_L[i]$ we check pairwise with every node in $C_S[j]$ to see if there are any overlaps in the adjacency lists. If there is an overlap with a node in $C_S[j]$ set $v_j$'s $alreadyClose=True$. After we are done checking with all subsections of $C_L[i]$, if all nodes $v_j$ represented by adjacency lists in $C_S$ have $alreadyClose==True$ we move on, but if any of them have $alreadyClose==False$ we return that the diameter is $>2$. This takes time $\sum_{i\in[1,k]} (\text{number subsections in } A^i_L)(\text{number subsections in } C_S)\left(\frac{M}{B}\right) = O\left(\frac{(|E|/M)^2 M}{B}\right) = O\left(\frac{|E|^2}{MB}\right)$.
	
	Now that we have verified all of the long adjacency lists we need only compare short lists to short lists. We do this by bringing in every pair for $i,j\in[0,\frac{8|E|}{M}]$ of $C_S[i]$ and $C_S[j]$ and checking if any two nodes are far from each other. This takes time $O\left(\frac{(|E|/M)^2M}{B}\right)= O\left(\frac{|E|^2}{MB}\right)$. 
	
	By summing the running time from the different parts of this algorithm, we obtain a total of $O\left(\frac{|E|^2}{MB} + sort(|E|)\right)$ I/Os in differentiating if the diameter of an undirected, unweighted graph is $1$, $2$, or greater than $2$. 
\end{proof}

Now that we have this framework we can give a cache-oblivious version of the algorithm. The proofs of cache-oblivious 2 vs 3 Diameter and 2 vs 3 Radius are included in Appendix \ref{sec:newUpperProofs}.

\subsection{Master Theorem in the I/O Model}
\label{sec:masterstheorem}
%A computation graph of a recurrence is a DAG representing the dependence of each problem on the subproblems that compose it. Given a level computation graph, if $M > \max(\min(w, h), 2\beta)$ where $h$ is the height of the level graph, $w$ is the width of the level graph, and $\beta$ is the degree, then the I/O complexity is upper bounded by the time it takes to read in the input for each parent node in the first level above the leaves (nodes without incoming edges) of the computation graph individually. Recurrences of the structure assumed in our theorem must take the form of a level computation graph where each level consists of problems of a particular size. 
In this section, we formally define our Master Theorem framework for the I/O model and provide bounds on the I/O complexity of problems whose I/O complexity fits the specifications of our framework. In addition, we also describe some example uses of our Master Theorem for the I/O model.  

The Master Theorem recurrence in the RAM model looks like $T(n)= a T(n/b)+f(n)$. We will use a similar recurrence but all functions will now be defined over $n$, $M$ and $B$.  
The I/O-Master Theorem function $f(n, M, B)$ includes all costs that are incurred in each layer of the recursive call. This includes the I/O complexity of reading in an input, processing the input, processing the output and writing out the output. 
In this section, we assume that $f(n, M, B)$ is a monotonically increasing function in terms of $n$ in order to apply our Master Theorem framework. What this means is that for any fixed $M$ and $B$ we want the number of I/Os to increase or stay the same as $n$ increases.  Given that $f(n, M, B)$ specifies the I/O complexity of reading in the inputs and writing out the outputs, we prove the following version of the Master Theorem in the I/O model.

\begin{theorem}[I/O Master Theorem]\label{thm:io-master-theorem}
	If $f(n,M,B)$ contains the cost of reading in the input (for each subproblem) and writing output (after computation of each subproblem), then the following holds. 
	Given a recurrence of $T(n,M,B) = \alpha T(n/\beta,M,B) + f(n,M,B)$, where $\alpha \geq 1$ and $\beta > 1$ are constants, and a base case of $T(n/x, M, B) = t(x, M, B)$ (where $t(x, M, B) = \Omega(1)$) for some $x \leq n$ and some function $t(x, M, B)$. Let $A(n,M,B) = \left(\frac{n}{x}\right)^{\log_{\beta}(\alpha)}t(x, M, B)$, $C(n, M, B) = \left(\frac{n}{x}\right)^{\log_{\beta}(\alpha) + \epsilon_1}t(x, M, B)$ and $D(n, M, B) = \left(\frac{n}{x}\right)^{\log_{\beta}(\alpha) - \epsilon_2}t(x, M, B)$ for some $\epsilon_1, \epsilon_2 > 0$, then we get the following cases:
	\begin{itemize}
		\item [Case 1:] If $f(n,M,B)  = O\left(D(n,M,B)\right)$ then $T(n) = \Theta\left(A(n,M,B) + \frac{n}{B}\right)$.
		\item [Case 2:] If $C(n,M,B)  = O\left(f(n,M,B)\right)$ and $\alpha T(n/\beta, M, B) \leq c f(n, M, B)$ for some constant $c< 1$ and all sufficiently large $n$, then $T(n) = \Theta\left(f(n,M,B) + \frac{n}{B}\right)$.
		\item [Case 3:] If $f(n, M, B) = \Theta(A(n, M, B))$, then $T(n, M, B) = \Theta(f(n, M, B)\log{\left(\frac{n}{x}\right)} + \frac{n}{B})$. 
		\item [Case 4:] If $f(n, M, B)$ has a constant number of terms, $f(n, M, B) =\Omega\left(\frac{n}{B}\right)$, and none of the previous cases are satisfied, then $T(n, M, B) = O\left(A(n,M,B)+f(n,M,B)\left(\frac{n}{x}\right)^{\log_{\beta}\alpha} + \frac{n}{B}\right)$ (note that this includes if $A$ and $f$ are incomparable), with tighter upper bounds provided in our proof (specifically Eqns.~\ref{eq:greater-one},~\ref{eq:equal-one}, and~\ref{eq:less-one}) depending on characteristics of the actual function, $f(n, M, B)$.
	\end{itemize}
\end{theorem}
\begin{proof}
First, note that, given our condition on $f$, we do not have to worry about clever maintaining of previous cache computations since $f(n, M, B)$ includes the cost of reading in input and writing out output. 

We first show that the final cost of the recurrence is

\begin{align}\label{eq:final-cost}
T(n, M, B) = F(n, M, B) + \Theta\left(\frac{n}{B}\right)
\end{align}

where the recursion cost of $F(n)$ is given as below

\begin{align}\label{eq:recursion-cost}
F(n, M, B) = \Theta(A(n, M, B)) + \sum_{j = 0}^{\log_{\beta}{(\frac{n}{x})} -1} \alpha^j f(n/\beta^j, M, B).
\end{align}

We consider the recursion tree of the recursion defined in Eq.~\ref{eq:recursion-cost}. The root is at level $0$. At level $j$ of the tree, there exists $\alpha^j$ nodes each of which costs $f(n/\beta^j, M, B)$ I/Os to compute. The leaves of the tree each cost $f(x, M, B) = t(x, M, B)$ time to compute. There exists $\Theta((\frac{n}{x})^{\log_{\beta} \alpha})$ leaves, resulting in a total cost of $\Theta((\frac{n}{x})^{\log_{\beta}{\alpha}}t(x, M, B))$ I/Os and the number of I/Os needed in the remaining nodes of the tree is $\sum_{j = 0}^{\log_{\beta}{(\frac{n}{x})} - 1}\alpha^j f(n/\beta^j, M, B)$. Summing these costs gives the recursion stated in Eq~\ref{eq:recursion-cost}. Finally, to obtain our final I/O cost given in Eq.~\ref{eq:final-cost}, we know that reading in the input incurs a fixed cost of $\Theta\left(\frac{n}{B}\right)$ I/Os regardless of the efficiency of the rest of the algorithm and the size of cache. 

Let $g(n, M, B) = \sum_{j = 0}^{\log_{\beta}{(\frac{n}{x})} - 1} \alpha^j f(n/\beta^j, M, B)$. We now bound $g(n, M, B)$:

\begin{enumerate}
\item [Case 1:] We prove $g(n, M, B) = O\left(\left(\frac{n}{x})^{\log_{\beta}{\alpha}}t(x, M, B\right)\right)$. Since we know that\\
$f(n, M, B) = O(\left(\frac{n}{x}\right)^{\log_{\beta}{\alpha} - \epsilon}t(x, M, B))$ for some $\epsilon > 0$ (and all $\epsilon_1 \leq \epsilon$), then we know that $f(n/\beta^j, M, B) = O\left(\left(\frac{n}{x\beta^j}\right)^{\log_{\beta}{\alpha}-\epsilon}t(x, M, B)\right)$. This then yields

\begin{align}\label{eq:case-1-gn}
g(n, M, B) = O\left(\sum_{j=0}^{\log_{\beta}{(\frac{n}{x})}-1} \alpha^j\left(\frac{n}{x\beta^j}\right)^{\log_{\beta}{\alpha}-\epsilon}t(x, M, B)\right)\\ = O\left(\left(\frac{n}{x}\right)^{\log_{\beta}{\alpha} - \epsilon_1} t(x, M, B)\left(\frac{\left(\frac{n}{x}\right)^{\epsilon_1}-1}{\beta^{\epsilon_1} - 1}\right)\right) = O\left(\left(\frac{n}{x}\right)^{\log_{\beta}{\alpha}}t(x, M, B)\right)
\end{align}

where we choose an arbitrarily small $\epsilon = \epsilon_1$ such that $\beta^{\epsilon_1} < 1$. 

\item [Case 2:] We prove $g(n, M, B) = \Theta(f(n, M, B))$. For sufficiently large $n$, we know that $\alpha(T(n/\beta, M, B)) \leq cf(n, M, B)$ for some constant $c< 1$. Let $c = c_0$ and $n'$ be the smallest constants such that this is satsified for all values of $M = [1, n^d]$ for constant $d$, $B = [L_B, n]$ where $L_B$ is the smallest value of $B$ that satisfies the algorithm's tall-cache assumption, and $n \geq n'$. 
Let $j = u$ be the largest exponent of $\beta^j$ such that $n/\beta^u \geq n'$. Therefore, we rewrite the equation for $g(n, M, B)$ in this case to be

\begin{align*}
g(n, M, B) &= \sum_{j=0}^{u-1} \alpha^j f\left(\frac{n}{\beta^j}, M, B\right) +\sum_{j=u}^{\log_{\beta}{\alpha} - 1} \alpha^j f\left(\frac{n}{\beta^j}, M, B\right) \\
&= \sum_{j=u}^{\log_{\beta}\alpha-1} \left(\frac{c_0}{\alpha}\right)^jf(n, M, B) + O(1)\\ 
&\leq f(n, M, B) \sum_{j = 0}^{\infty} \left(\frac{c_0}{\alpha}\right)^j + O(1) \\
&\leq f(n, M, B) \left(\frac{1}{1-c_0}\right) + O(1) = O(f(n, M, B)). 
\end{align*}

Trivially (by the case when $j = 0$), we know that $g(n, M, B) = \Omega(f(n, M, B))$. 

\item [Case 3:] We prove $g(n, M, B) = \Theta\left(\left(\frac{n}{x}\right)^{\log_{\beta}{\alpha}} t(x, M, B)\log{\left(\frac{n}{x}\right)}\right)$. Since we know that \\
$f(n) = \Theta\left(\left(\frac{n}{x}\right)^{\log_{\beta}{\alpha}}t(x, M, B)\right)$, we then also know that 
$$f(n/\beta^j, M, B) = \Theta\left(\left(\frac{n}{x\beta^j}\right)^{\log_{\beta}{\alpha}}t(x, M, B)\right)$$
 and thus obtain the following $g(n, M, B)$ in this case

\begin{align*}
g(n, M, B) = \Theta\left(\sum_{j=0}^{\log_{\beta}\left(\frac{n}{x}\right)-1} \alpha^j \left(\frac{n}{x\beta^j}\right)^{\log_{\beta}{\alpha}}t(x, M, B)\right) = \Theta\left(\left(\frac{n}{x}\right)^{\log_{\beta}{\alpha}}t(x, M, B) \log_{\beta}{\left(\frac{n}{x}\right)}\right).
\end{align*} 

\end{enumerate}

We now prove the first three cases using our bounds on $g(n, M, B)$ above.

\begin{enumerate}
\item [Case 1:] $F(n, M, B) = \Theta(A(n, M, B)) + O\left(\left(\frac{n}{x}\right)^{\log_{\beta}{\alpha}}t(x, M, B)\right) = \Theta(A(n, M, B))$\\ $T(n, M, B) = \Theta\left(A(n, M, B) + \frac{n}{B}\right)$
\item [Case 2:] $F(n, M, B) = \Theta(A(n, M, B)) + \Theta(f(n, M, B)) = \Theta(f(n, M, B))$\\ $T(n, M, B) = \Theta\left(f(n, M, B) + \frac{n}{B}\right)$
\item [Case 3:] $F(n, M, B) = \Theta(A(n, M, B)) + \Theta\left(\left(\frac{n}{x}\right)^{\log_{\beta}{\alpha}}t(x, M, B) \log_{\beta}{\left(\frac{n}{x}\right)}\right) = \Theta\left(A(n, M, B) \log{\left(\frac{n}{x}\right)}\right)$\\ $T(n, M, B) = \Theta\left(A(n, M, B)\log\left(\frac{n}{x}\right) + \frac{n}{B} \right)$
\item [Case 4:] Given our assumption that $f(n, M, B) = \Omega(n/B)$, we first need to show an upper bound on $g(n, M, B)$ (it is trivially $\Omega(f(n, M, B))$). 

By using Eq.~\ref{eq:case-1-gn}, we can obtain the following upper bound on $g(n, M, B)$ when $\frac{\alpha}{f'(\beta)} > 1$, where $f'(\beta)$ is the largest expression containing $\beta$ in the dominator produced by $f(n/\beta, M, B)$ or $1$ if $f'(\beta) = 0$. For example, if $f(n/\beta, M, B) = (n/\beta)^2$, then $f'(\beta) = \beta^2$ or if $f(n/\beta, M, B) = (\beta n)^2 + \frac{1}{\beta}$, then $f'(\beta) = 1/\beta$.

\begin{align}\label{eq:top}
g(n, M, B) = \sum_{j = 0}^{\log_{\beta}{\left(\frac{n}{x}\right)-1}} \alpha^j f(n/\beta^j, M, B) &\leq \left(\frac{\frac{\alpha}{f'(\beta)}}{\frac{\alpha}{f'(\beta)}-1}\right)  \left(\frac{n}{x}\right)^{\log_{\beta}{\left(\frac{\alpha}{f'(\beta)}\right)}} f(n, M, B) \\ &= O\left(f(n, M, B) \left(\frac{n}{x}\right)^{\log_{\beta}{\left(\frac{\alpha}{f'(\beta)}\right)}}\right).
\end{align}

Thus, this gives us the final expression for $T(n)$ to be

\begin{align}
T(n) &= \Theta(A(n, M, B)) +  O\left(f(n, M, B) \left(\frac{n}{x}\right)^{\log_{\beta}{\left(\frac{\alpha}{f'(\beta)}\right)}}\right) + \Theta\left(\frac{n}{B}\right)\\ 
&= O\left(A(n, M, B) +  f(n, M, B) \left(\frac{n}{x}\right)^{\log_{\beta}{\left(\frac{\alpha}{f'(\beta)}\right)}} + f(n, M, B) + \frac{n}{B}\right).\label{eq:greater-one}
\end{align} 

If $\frac{\alpha}{f'(\beta)} = 1$, then 

\begin{align}\label{eq:equal-one}
T(n) = \Theta\left(A(n, M, B) + f(n, M, B)\log\left(\frac{n}{x}\right) + f(n, M, B) + \frac{n}{B}\right).
\end{align}

Finally, if $\frac{\alpha}{f'(\beta)}<1$, then we obtain the following expression for $g(n, M, B)$:

\begin{align*}
g(n, M, B) = \sum_{j = 0}^{\log_{\beta}\left(\frac{n}{x}\right) - 1} \alpha^j f(n/\beta^j, M, B) \leq \left(\frac{1}{1 - \frac{\alpha}{f'(\beta)}}\right) f(n, M, B) .
\end{align*} 

Trivially, $g(n, M, B) = \Omega(f(n, M, B))$. Therefore, we know that the final expression for $T(n)$ to be

\begin{align}\label{eq:less-one}
T(n) = \Theta\left(A(n, M, B) + f(n, M, B) + \frac{n}{B}\right). 
\end{align}
\end{enumerate}
\end{proof}

Note that we do not present the proofs for when we need to take $\floor{\log_{\beta}\left(\frac{n}{x}\right)-1}$ or $\ceil{\log_{\beta}\left(\frac{n}{x}\right)-1}$ since the proofs are nearly identical to that presented for the original Master Theorem~\cite{CLRS09}.

\paragraph{One-Layer Self-Reductions}

We state a relationship between one-layer self-reductions and our Master Theorem framework above. We refer to the process of solving a problem by reducing to several problems of smaller size each of which can be solved in cache and one recursive call is necessary as a \emph{one-layer self-reduction}. Suppose the runtime of an algorithm in the RAM model is $n^{\log_{\beta}{\alpha}}$, then by dividing the problems into $\frac{n^{\log_{\beta}{\alpha}}}{M}$ subproblems each of which takes $M/B$ I/Os to process, the I/O complexity of the algorithm is $\Theta\left(\frac{n^{\log_{\beta}{\alpha}}}{M^{\log_{\beta}\alpha -1}B}\right)$ which is the same result we obtain via our Master Theorem framwork above when $t(x, M, B) = M/B$. 

We now prove formally the theorem related to one-layer self-reductions. 

\begin{theorem}
Let $P$ be a problem of size $n$ which can be reduced to $g(n/M)$ sub-problems, each of which takes $T(M, M, B)$ I/Os to process. %Let $g(n/M)$ be the number of instance of size $M$ that take $T(M, M, B)$ time to process. 
The runtime of such a one-layer self reduction for the problem $P$ is $T(n, M, B) = g(n/M) T(M, M, B) + f(n, M, B)$ where $f(n, M, B) = \Omega\left(\frac{n}{B}\right)$. 
\end{theorem}

\begin{proof}
The number of I/Os needed to process a subproblem of size $M$ is given by $T(M, M, B)$. If $g(n/M)$ is the total number of such subproblems of size $M$ that need to be processed, then the total number of I/Os needed to process all subproblems is $O(g(n/M)T(M, M, B))$. Then, to read in the input of size $n$ requires $\Omega\left(\frac{n}{B}\right)$ I/Os. Any other I/Os incurred while processing the $g(n/M)$ would result in a total of $f(n, M, B)$ I/Os.
\end{proof}

%\xxx{TODO: give statements about if you think RAM model time is $n^{\lg_b(a)}$ and you think a self-reduction/recurence describes your best running time then we get the $n^{\lg_b(a)}/(M^{\lg_b(a)-1}B)$ time lower bound }

%\subsection{Runtimes of Standard I/O-Model Algorithms}\label{sec:mt-proofs}
%
%A reduction between APSP and a problem known as $(\min, +)$-matrix multiplication has been known since the 1990s~\cite{seidel92}. We present a fine-grained reduction in the I/O model between APSP and $(\min, +)$-matrix multiplication in Section~\ref{sec:APSPReductions}. Therefore, solving $(\min, +)$-matrix multiplication also solves APSP in the I/O model.
%
%\begin{theorem}\label{thm:apsp-mt}
%The APSP problem has runtime...\xxx{fill in and check recurrence}
%\end{theorem}
%
%\begin{theorem}\label{thm:3sum-sr}
%The $3$-SUM problem can be solved via a self-reduction in \xxx{fill in}
%\end{theorem}
%
%\begin{theorem}\label{thm:ov-sr}
%The orthogonal vectors problem can be solved via a self-reduction in \xxx{fill in}
%\end{theorem}
%
%\begin{theorem}\label{thm:lcs-mt}
%The longest common subsequence recursion from~\cite{ChowdhuryRa06} gives \xxx{fill in and check recurrence}
%\end{theorem}

\subsection{Faster I/O Matrix Multiplication via I/O Master Theorem}
\label{sec:MMspeed}
As we mentioned above, any I/O algorithm that has a self-reduction to one of the forms stated in Section~\ref{sec:masterstheorem}. Using our I/O Master Theorem, we can show a comparable I/O matrix multiplication bound to the matrix multiplication bound based on finding the rank of the Matrix Multiplication Tensor in the RAM model. 

Recent improvements to matrix multiplication's running time also imply faster cache oblivious algorithms. Recent work has improved the bounds on $\omega$ where $\omega$ is the constant such that for any $0<\epsilon<1$ there is an algorithm for $n$ by $n$ matrix multiplication that runs in $n^{\omega+\epsilon}$. The most recent improvements on these bounds have been achieved by bounding the rank of the Matrix Multiplication Tensor~\cite{vstoc12,legall}. 

The I/O literature does not seem to have kept pace with these improvements. While previous work discusses the efficiency of naive matrix multiplication and Strassen matrix multiplication, it does not discuss the further improvements that have been generated. 

We note in this section that the modern techniques to improve matrix multiplication running time, those of bounding the rank of the Matrix Multiplication Tensor, all imply cache-efficient algorithms.

\begin{theorem}[Matrix Multiplication I/O Complexity\cite{COPPERSMITH1990251}]\label{thm:mm-self-reduction}
	Let $T_{MA}(n) = O\left(\frac{n^2}{B}\right)$ be the time it takes to do matrix addition on matrices of size $n$ by $n$. 
If the matrix multiplication tensor's rank is bounded such that the RAM model running time is $n^{\omega'+\epsilon}$ for any $0<\epsilon<1$ then the following self-reduction exists for some constant $\alpha$,

$$T_{MM}(n) = O\left(\alpha^{\omega '+\epsilon}T_{MM}\left(\frac{n}{\alpha}\right) +  \alpha^{\omega '}T_{MA}\left(\frac{n}{\alpha}\right)\right).$$ 
\end{theorem}

Self-reductions feed conveniently into cache oblivious algorithms. Notably, when we plug this equation into the I/O Master Theorem from Section~\ref{sec:masterstheorem}, we obtain the following bound as given in Lemma~\ref{lem:mt-mm}. Recursive structures like this tend to result in cache-oblivious algorithms. After all, regardless of the size of cache, the problems will be broken down until they fit in cache. Then, when a problem and the algorithm's execution fit in memory, the time to answer the query is $O(M/B)$, regardless of the size of $M$ and $B$. 

\begin{lemma}\label{lem:mt-mm}
If the matrix multiplication tensor's rank is bounded such that the RAM model running time is $n^{\omega'+\epsilon}$ for any $0<\epsilon<1$ and Theorem~\ref{thm:mm-self-reduction} holds with base case $T_{MM}(\sqrt{M}) = \frac{M}{B}$, then the running time for cache-oblivious matrix multiplication in the I/O model is at most $O\left(\frac{n^{\omega'+\epsilon}}{M^{\frac{\omega'+\epsilon}{2}-1}B} + \frac{n^2}{B}\right)$ for any $0<\epsilon<1$.
\end{lemma}
\begin{proof}
The algorithm that uses the self-reduction implied by bounding the border rank produces the recurrence $T_{MM}(n) = O\left(\alpha^{\omega '+\epsilon}T_{MM}\left(\frac{n}{\alpha}\right) +  \alpha^{\omega '}T_{MA}\left(\frac{n}{\alpha}\right)\right)$ by Theorem~\ref{thm:mm-self-reduction}. By our bound on $T_{MA}(n) = O\left(\frac{n^2}{B}\right)$ and our base case $T_{MM}\left(\sqrt{M}\right) =  \frac{M}{B}$, we use Case $4$ of our I/O Master Theorem (Theorem~\ref{thm:io-master-theorem}) from Section~\ref{sec:masterstheorem} to get $T_{MM}(n) = O\left(\frac{n^{\omega'+\epsilon}}{M^{\frac{\omega'+\epsilon}{2}-1}B} + \frac{n^2}{B}\right)$ since $\frac{\alpha^{\omega' + \epsilon}}{f'(\alpha)} > 1$. 

The algorithm implied by the self-reduction does not change depending on the size of the cache or the cache line. The running time will follow this form regardless of $M$ and $B$. Thus, this algorithm is cache-oblivious. 

\end{proof}

\subsection{RAM  Reductions Imply I/O Upper Bounds}
\label{sec:Hy1}

Reductions are generally between problems of the same running time in terms of their input size. There are also reductions showing problems that run faster in their running time are harder than problems that run slower. For example if zero triangle is solved in $(|E|)^{1.5-\epsilon}$ then 3-SUM is solved in $O(n^{2-\epsilon'})$ time.
An open problem in fine-grained complexity is showing that problem that runs slower is harder than a problem that runs faster in a reductions sense. Some types of reductions in the inverse direction would imply a faster I/O algorithm for 0 triangle and thus imply faster I/O algorithms for APSP and (min,+) matrix multiplication. 

Notably 3-SUM saves a factor of $MB$ in the I/O-model whereas the 0 triangle problem saves a factor of only $\sqrt{M}B$. 

One type of reduction in the RAM model is of the form 
$$T_{0 \triangle}(n) = O \left(\sum_{s \in S}T_{3SUM}(s) + n^{3-\epsilon}\right).$$
Where the running times of the $T_{3SUM}$ problems summed equal $n^{3-\epsilon}$ if 3-SUM is solvable in truly sub-quadratic time. In this case we will get an improvement over the zero triangle running time when $MB = O(n^{\epsilon})$. Notably, if the extra work done is I/O-efficient then the zero triangle problem could be solved faster at a wider range of values of $M$ and $B$. The reductions we have covered in this paper have had the extra work be efficient. The reductions are efficient in spite of the fact that these reductions were originally RAM model reductions which did not care about memory locality. 

There is, however, one kind of RAM reduction which does not imply speedups in the I/O-model. Following in the style of Patrascu's convolution 3-SUM reduction \cite{patrascu2010towards}, we can have a reduction of the form 
\begin{align}
T_{0 \triangle}(n) = O(g^cn^{3-2i}T_{3SUM}(n^i/g) + n^{3}/g),\label{eq:0-triangle}
\end{align}
where $g$ is an integer, $c$ is an arbitrary constant and $i\in[0,3/2]$. These reductions imply speedups when a polynomial time improvement is made for 3SUM, but does not immediately imply speedups if no polynomial time improvement is made for 3SUM. If the additional $n^3/g$ work is I/O-inefficient, this reduction might not imply speedups in the I/O model. 

If one is trying to show hardness for 3SUM from APSP, the approach that does not imply algorithmic improvements must have a large amount of I/O-inefficient work. We suggest the more fruitful reductions to look for have the form of Eq.~\ref{eq:0-triangle}.

\section{Novel Reductions}
In this section we cover reductions related to Wiener Index, Median, Single Source Shortest Paths, and s-t Shortest Paths. We first cover our super linear lower bounds, then cover the linear lower bounds. 

\subsection{Super Linear Lower Bounds}
\label{sec:newLB}
We present reductions in the I/O model which yield new lower bounds. We have as corollaries of these same reductions related lower bounds in the RAM model.  Many of these reductions relate to the problem of finding the Wiener Index of the graph. 

We show diameter reduces to Wiener Index, APSP reduces to Wiener Index, and we show 3 vs 4 radius reduces to median.

\begin{definition}[Wiener Index]
Given a graph $G = (V, E)$ let $D[i,j]$ be the shortest path distance between node $i$ and node $j$ in the graph $G$, where $D[i, i] = 0$), $n = |V|$, and $m = |E|$.
The Wiener index of the graph is $\sum_{\forall i\in V}\sum_{\forall j \in V} D[i,j]$.
\end{definition}

The Wiener Index measures the total shortest path distance between all pairs of notes. Intuitively this is a measure of graph distance. Much like median, radius and diameter. If lots of nodes are far apart then the Wiener index will be high, if the nodes are close together then the Wiener index will be low. 

\begin{lemma}
	If (directed/undirected) Wiener index in a (weighted/unweighted) graph, $G = (V, E)$, is solvable in $T(n, m, M,B)$ time then for any choice of sets $X \subset V$ and $T \subset V$ the sum $\sum_{x \in X}\sum_{t \in T} \delta(x,t)$ is computable in $O\left(T(n,m,M,B)+\frac{m}{B}\right)$ time. 
	\label{lem:weinerSubsets}
\end{lemma}
\begin{proof}
	Create a graph $G'$ by adding two special sets of nodes $X'$ and $T'$ to $G$. Specifically, $\forall x \in X$ add $x'$ to $G'$ and add an edge with weight 1 (or if $G$ is an unweighted graph just an edge) to the node $x$ in $G$. And, $\forall t \in T$ add $t'$ to $G'$ and add an edge with weight 1 (or if $G$ is an unweighted graph just an edge) to the node $t$ in $G$. 
	
	Now we will ask for the Wiener index of 4 graphs $G'$, $G+X'$, $G+T'$ and $G$. Let $WI(G)$ be the Wiener index of graph G. Note that the shortest path between $x'$ and $t'$ in $G'$ is $\delta(x,t)+2$ and the shortest path between $x_1',  x_2' \in X'$ is $\delta(x_1,x_2)+2$ and these paths always use edges $\{(x',x), (t,t')\}$, $\{(x_1', x_1), (x_2', x_2)\}$, respectively, and otherwise exclusively edges in $G$. Thus, we have that the formula
	$$W(G, G') = WI(G')-WI(G+X')-WI(G+T')+WI(G) =  \sum_{x' \in X'} \sum_{t' \in T'} d_{G'}(x',t')=  \sum_{x \in X} \sum_{t \in T} d_{G}(x,t)+2$$
		
	So if we return $-2|X'||T'| + W(G, G') = -2|X'||T'|+\sum_{x' \in X'} \sum_{t' \in T'} d_{G'}(x',t')= -2|X'||T'|+ \sum_{x \in X} \sum_{t \in T'} d_{G}(x,t)+2 = \sum_{x \in X} \sum_{t \in T'} d_{G}(x,t)$, which is the desired value. We need to add $O(n)$ nodes and edges to the original graph which costs $O\left(\frac{m}{B}\right)$ time. Then, we need to run the algorithm for Wiener index four times resulting in a total I/O complexity of $O\left(T(n, m, M, B) + \frac{m}{B}\right)$.
\end{proof}

\begin{lemma}
	If undirected Wiener index in an unweighted graph, $G = (V, E)$, is solvable in $T(n, m, M, B)$ time then for any choice of sets $X \subset V$ and $T \subset V$ the sum\\ $\sum_{x \in X}\sum_{t \in T}  \max\{\min\{\delta(u,v),k+1\},k\}$ is computable in $O\left(T(kn,km,M,B)+\frac{k|E|}{B}\right)$ time. 
	\label{lem:weinerHelper}
\end{lemma}
\begin{proof}
	Given an undirected graph, $G = (V, E)$, replace all edges with two directed edges (that form a cycle between the two endpoints of the original edge) and then proceed as described below with the new directed graph. 
	
	We will generate a new graph $G'$ with sections $G_1, G_2, \ldots, G_{k+1}$ and nodes $s_1, \ldots, s_k$ in the following way. Given a directed graph $G$ make $k+1$ copies of the vertex set, call them $G_1, G_2, \ldots, G_{k+1}$. Let $v_i \in V_i$ be the copy of $v\in V$ in $G_i$. Add an edge between $v_i$ and $v_{i+1}$. Add an edge between $u_i$ and $v_{i+1}$ if an edge exists from $u$ to $v$ in the initial edge set $E$. Next we add extra structure by adding $k$ nodes $s_1, s_2, \ldots , s_k $, where node $s_i$ connects to every node in $G_i$ via edges $(v_i, s_i)$ and every node in $G_{i+1}$ via edges $(s_i, v_{i+1})$. 
	
	Now note that the distance between $v_1$ and $u_{k+1}$ is at least $k$.  If any $s_i$ is ever used, then it is guaranteed that any path from $v_1$ to $u_{k+1}$ uses at least $k+1$ edges. Otherwise, all paths between nodes in the first layer and $k+1$-st layers require $k$ edges. 
	The paths can't be longer than $k+1$ because $v_1 \rightarrow s_1 \rightarrow s_2 \rightarrow \ldots \rightarrow s_k \rightarrow u_{k+1}$ is a $k+1$ length path that always exists between every pair of nodes in the first layer and the $k+1$-st layer. 
	
	A path of length $k$ exists in this new graph iff a path of length $k$ exists in $G$ from $u$ to $v$. STherefore, the distance between $u_1$ and $v_{k+1}$ equals $\max\{\min\{\delta(u,v),k+1\},k\}$. 
	Now we can use Lemma \ref{lem:weinerSubsets} to compute the sum  $\sum_{x \in X}\sum_{t \in T}  \delta_{G'}(x_1, t_{k+1})$. As discussed $\delta_{G'}(x_1, t_{k+1})=\max\{\min\{\delta(x,t),k+1\},k\}$. Thus, we get the desired sum. 
	
	\begin{figure}[ht]
		\centering
		\includegraphics[width=0.6\textwidth]{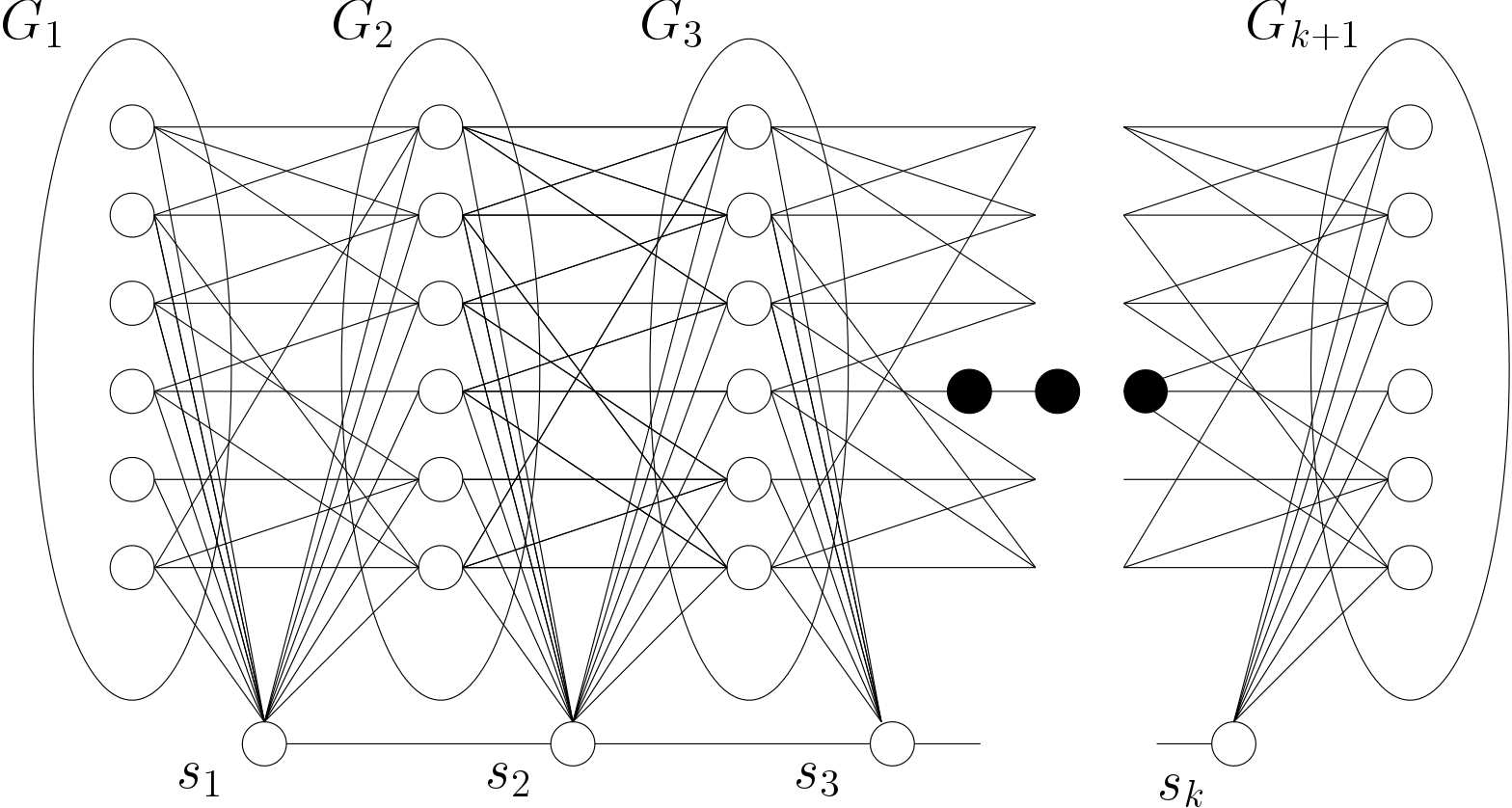}
		\caption{Depiction of the graph used for reduction from counting nodes at distance $k$ to the Wiener index. }
		\label{fig:weinerDistances}
	\end{figure}

\end{proof}

\begin{theorem}
If undirected Wiener index in an unweighted graph is solvable in $T(n, m, M, B)$ time then counting the number of pairs $x \in X$ and $t \in T$ in a directed/undirected, unweighted graph where $\delta(x,t) = k$ and computing the number of pairs where $\delta(x,t) \geq k$ is computable in $O\left(T(kn,km,M,B)+\frac{km}{B}\right)$ time. 
\label{thm:weinerDistacncek}
\end{theorem}
\begin{proof}
We use Lemma \ref{lem:weinerHelper} to compute $r_k = \sum_{x \in X}\sum_{t \in T}  \max\{\min\{\delta(u,v),k\},k-1\}$. Let $a_k$ be the number of pairs $(x,t)$ where $\delta(u,v)\geq k$ then  $r_k-|X||T|(k-1) = a_k$. 
Note that $a_k - a_{k+1}$ is the number of nodes at exactly distance $k$. So with two calls to the algorithm from Lemma~\ref{lem:weinerHelper} we can compute the number of nodes at distance $k$. 
We can now return $a_k-a_{k+1}$ and $a_k$ and get both values. 
\end{proof}

We now show that Wiener index, in sparse graphs, can efficiently return small diameters. Notably, this means that improvements to the sparse Wiener index algorithm will imply faster algorithms for the sparse diameter problem than exist right now. 

\begin{corollary}
If Wiener index is solvable in $T(n,m,M,B)$ time then returning  $\min \{\text{diameter}, k\}$ is solvable in $O\left(\log(k) T(kn,km, kM, kB) + \frac{km}{B}\right)$ time. 
\end{corollary}
\begin{proof}
Using Theorem \ref{thm:weinerDistacncek} we can binary search to find the largest value such that there is a node at distance $d$ and there are no nodes larger than $d$. If this value is above $k$ then there will be nodes at distance $k+1$ which is efficient to check.  
\end{proof}

\begin{corollary}
If Wiener index is solvable in $T(n,m,M,B)$ time then returning the number of nodes at distances in $[1,k]$ from each other can be done in $O\left(k T(kn,kM,kB) + \frac{k^2m}{B}\right)$ time. 
\end{corollary}
\begin{proof}
	Using Theorem \ref{thm:weinerDistacncek} we can check every distance from $1$ up to $k$ and return the number of nodes at that distance. 
\end{proof}

Next, we prove that improvements to median finding in sparse graphs improve the radius algorithm, using a novel reduction. Notably, in the I/O model 3 vs 4 radius is slower than 2 vs 3 radius; whereas, in the RAM-model, these two problems both run in $n^2$ time. The gap in the I/O model of a factor of $M$ is what allows us to make these statements meaningful. 

\begin{theorem}\label{thm:3v4radius-to-median}
	If median is solvable in $T(n,m, M,B)$ time then 3 vs 4 radius is solvable in\\
	$O\left(T(n, m, M, B)+\frac{n^2}{MB} + \frac{E\log(E)}{B}\right)$ time. 
\end{theorem}
\begin{proof}

	First we run the algorithm from Theorem \ref{thm:radiusCache} to determine if the radius is $\leq 2$. If the radius is $\geq 3$ then we will produce $G'$ by doing the following. Running this algorithm takes $O\left(\frac{n^2}{MB} + \frac{E\log(E)}{B}\right)$ time.

	See Figure \ref{fig:3vs4Median} for an image of the completed $G'$.
	
		\begin{figure}[ht]
		\centering
		\includegraphics[width=0.6\textwidth]{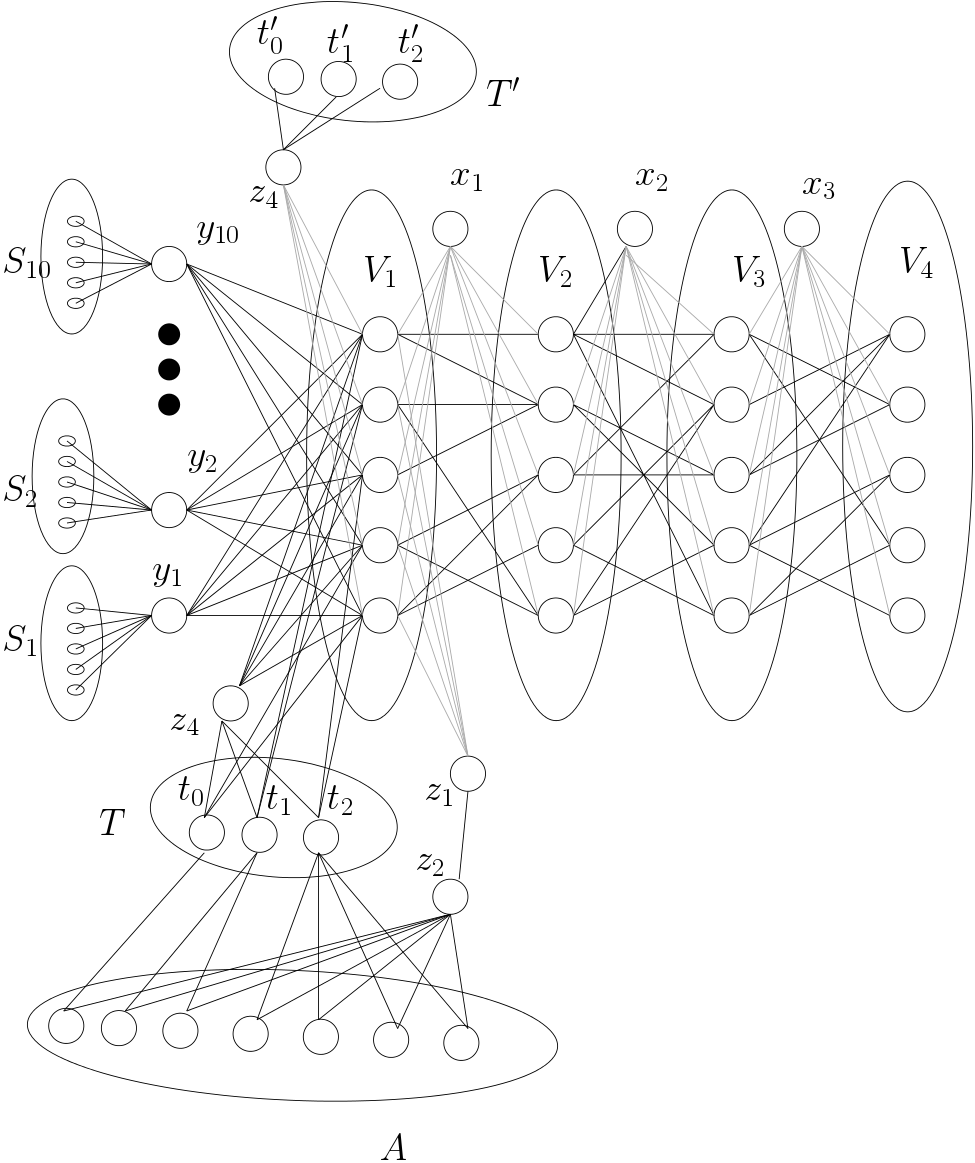}
		\caption{Depiction of the completed $G'$ graph from the 3 vs 4 radius to median reduction. Some lines are gray for readability. All lines, both gray and black, are the same undirected unweighted edges. }
		\label{fig:3vs4Median}
	\end{figure}	
	
	We start by adding four copies of the vertex set $V_1, \ldots , V_4$ to $G'$ and add edges between nodes in $v \in V_i$ and $u \in  V_{i+1}$ if $v$ and $u$ are connected in the original graph $G$. Additionally add nodes $x_1,\ldots, x_3$ where $x_i$ is connected to all nodes in $V_i$ and all nodes in $V_{i+1}$. 
	
	We want to enforce that a node in $V_1$ is the median, we can do this by adding many nodes close to nodes in $V_1$ and far from other nodes. We then want to check if there is a node in $V_1$ that is at distance $3$ from all nodes in $V_4$. The median will give us the node with the smallest total distance so we will want to correct for how close nodes in $V_1$ are to nodes in $V_2$ and $V_3$. 
	
	First we add 10 nodes $y_1, \ldots, y_{10}$ each connected to every node in $V_1$. Next add $S_1,\ldots, S_{10}$ which are sets of $10n$ nodes, node $y_i$ will connect to all nodes in $S_i$. There are less than $7n$ nodes in the rest of the graph. Nodes in $V_1$ are at distance $2$ from all the nodes in $S_i$. Nodes in $S_i$ are at distance $\geq 3$ from all nodes except those in $V_1$, $y_i$ and $S_i$. So nodes in $V_1$ are the only possible medians, as they are closer by at least $90n$ to all nodes in the various $S_i$. 
	
	The algorithm we gave for 2 vs 3 radius (Theorem \ref{thm:radiusCache}) can return the number of nodes at distance 0, 1, 2, $\geq 3$ for each nodes. We run this algorithm on $U = V_1 \cup x_1 \cup V_2 \cup x_2 \cup V_3$ and keep track of all of these numbers for each node. We will then add extra structure $W$ so that all nodes in $V_1$ will have the same total sum of distances to all nodes in $G'-V_4$. Then, the median returned will be whatever node in $V_1$ that has a minimum distance to $V_4$. Now, we note that if the radius is $\leq 3$ the distance from a node to $V_4$ will be $3n$. Next, we describe how to build $W$.
	
	The total distance to all nodes in $U$ from any node in $V_1$ will be between $4n+5$ and $6n+5$. So we create a node set $A$ of $2n$ nodes to which we then add $z_1$ and $z_2$. We connect $z_1$ to all of $V_1$ and $z_2$. We connect $z_2$ to $A$. So, every node in $A$ is at distance $\leq 3$ from $V_1$. We then add $\log(2n)$ nodes $T=t_0,t_1,\ldots,t_{\log(n)}$. We connect node $t_i$ to $2^i$ nodes which are non-overlapping with any of the other $t_j$. We add a final node, $z_3$, which connects to the nodes in $V_1$ and all nodes in $t_j$. We can now connect a node in $V_1$ to $t_j$ to make its total sum of distances $2^j$ smaller. We use this to equalize all the sums of distances. We finally need to add another $\log(2n)$ nodes $T' = \{t'_0, \dots, t'_{\log(n)}\}$ and a node $z_4$. We connect $z_4$ to all of $V_1$ and all of $T'$. We connect a node in $V_1$ to $t'_j$ if it isn't connected to $t_j$. Now all nodes in $V_1$ have the same sum of distances to all nodes in $G'-V_4$. 
	
	We now run median. If the node returned has distance $3n$ to the nodes in $V_4$ then the radius is $3$. Otherwise, the radius is $4$. 
	
\end{proof}

It has previously been shown that Wiener Index is equivalent to APSP in the RAM model. Here we show this also holds in the I/O-model. 

\begin{theorem}
	If Wiener Index is solvable in $n^{3-\epsilon}$ time in a dense graph then APSP is solvable in $\tilde{O}(n^{3-\epsilon}+n^2/B)$ time. 
	\label{thm:apsp-lb-wi}
\end{theorem}
\begin{proof}
	We will show that Wiener Index solves the negative triangle problem. We start by taking the tripartite negative triangle instance with weights between $-W$ and $W$ and calling the partitions $A$, $B$ and $C$ respectively. 
	
	We create a new graph $G'$ and add to it $A$, $B$, $C$, and $A'$. Where $A'$ is a copy of $A$. We add edges between $A$ and $B$, $B$ and $C$, $C$ and $A'$ if they existed in the origional graph.  We put weights on those edges equal to their weight in the original graph plus $5W$. 
	
	Now if there is a node $a\in A$ involved in a negative triangle then the distance from $a$ to $a'$ in $G'$ will be $<15W$. 
	
	We add a node $x$ which is connected to every node in $G'$ with an edge of length $15W$, this ensures all nodes in the graph have a path of length $30W$ to each other in the worst case. This guarentees that no nodes have infinite path lengths to each other which would result in an un-useful response from the Wiener index. 
	
	Finally we add edges between $A$ and $A'$. Specifically, between two copies of the same node say $a$ and $a'$ we add an edge of length $15W$. Between two nodes that aren't copies, say $a$ and $v'$ we add an edge of length $12W$. 
	
	Note that the distance from $a$ to $a'$ is $<15W$ only if there is a negative triangle through $a$. Any path through $x$ that doesn't begin or end at $x$ must have length $>30W$. A path through $A$ then $B$ then $C$ then $A'$ represents a triangle. A path from $a$ to $a'$ that goes through some $v'$ will have length at least $12W+4W+4W=20W$. So, any short paths between nodes represents a negative triangle. 
	If the smallest triangle through $a$ in the original graph had total weight positive or zero, then the distance from $a$ to $a'$ will be $15W$, using the edge we added between them. 
	
	The shortest path from $a\in A$ to $v'\in A'$ when the nodes are not copies of each other is $12W$. There is an edge between them of this length, so it can be no longer. Using $x$ is less efficient and the shortest path from $A$ to $B$ to $C$ to $A'$ will have weight at least $4W$ on each of the three edges, thus be at least $12W$ in distance. 
	
	Now, if there are no negative triangles the total sum of weights between $A$ and $A'$ will be $2n(12n+3)W$. If there is a negative triangle, then at least one of the pairs will have total length $<15W$ causing the sum to be strictly less than $2n(12n+3)W$.
	
	We use Lemma \ref{lem:weinerSubsets} to find the total sum of distances between $A$ and $A'$. 
	
	The time to make these copies and add edges is $n^2/B$ time.
\end{proof}

%%%%%% More results but they reference not-perviously-defined problems and are thus probably best left to another later paper

%\xxx{TODO: This part needs to be written up}

%begin{theorem}
%	If Wiener List is solvable in $T(n,M,B)$ time then k vs k+1 radius is solvable in $O(T(kn,kM,kB))$ time. 
%\end{theorem}

%\begin{theorem}
%	If Wiener List is solvable in $T(n,M,B)$ time then eccentricities in graphs with diameter $k$ is solvable in $O(kT(kn,kM,kB))$ time. 
%\end{theorem}

%\begin{theorem}
%	If maxian is solvable in $T(n,M,B)$ time then 3 vs 4 diameter is solvable in $O(T(n,M,B))$ time. 
%\end{theorem}\begin{proof}
%\xxx{TODO: given that our algorithm for 2 vs 3 diameter actually counts the number of nodes at distance 0, 1, and 2 from every node then we can make a 4 layer graph and compute the distances to the first 3 layers and punish each node so that every node has a total sum of distances to the first three layers  plus the extra structure of the same weight. Then when we ask maxian we are getting the node at maximum distance from all the nodes in layer 4. If this node is at distance 3 from any of these nodes then we have found a node with diameter 4. }
%\end{proof}
\newcommand{\GirthE}{Girth-Containing-Edge }
 \newcommand{\GirthV}{Girth-Containing-Vertex }
 \newcommand{\stsp}{$s$-$t$-shortest-path }

\subsection{Linear-Time Reductions}
\label{sec:linear}

In fine-grained complexity, it often does not make sense to reduce linear-time problems to one another because problems often have a trivial lower bound of $\Omega(n)$ needed to read in the entire problem. However, in the I/O model, truly linear time---the time needed to read in the input---is $\Theta(n/B)$. Despite significant effort, many problems do not achieve this full factor of $B$ in savings, and thus linear lower bounds of $\Omega(n)$ are actually interesting. We can use techniques from fine-grained complexity to try to understand some of this difficulty. 

In the remainder of this section, we cover reductions between linear-time graph problems whose best known algorithms take longer than $O(|E|/B)$ time. This covers many of even the most basic problems, like the $s$-$t$ shortest path problem, which asks for the distance between two specified nodes $s$ and $t$ in a graph $G$. The \emph{sparse}  $s$-$t$ shortest paths problem has resisted improvement beyond $O(|V|/\sqrt{B})$ even in undirected unweighted graphs \cite{Mehlhorn2002}, and in directed graphs, the best known algorithms still run in time $\Omega(|V|)$ \cite{chiang1995external,ABDHI07}.  

Notably, the undirected unweighted $s$-$t$ shortest path problem is solved by Single Source Shortest Paths (SSSP) and Breadth First Search (BFS). Further note that for directed graphs the best known algorithms for SSSP, BFS, and Depth First Search (DFS) in sparse, when $|E| = O(|V|)$, directed graphs take $O(|V|)$ time. Which is a cache miss for every constant number of operations, giving no speed up at all. SSSP, BFS, and DFS solve many other basic problems like graph connectivity. 

By noting these reductions we want to show that improvements in one problem propagate to others. We also seek to explain why improvements are so difficult on these problems. Because, improving one of these problems would improve many others, any problem which requires new techniques to improve implies the others must also need these new techniques. Furthermore, any lower bound proved for one problem will imply lower bounds for the other problems reduced to it. We hope that improvements will be made to algorithms or lower bounds and propagated accordingly. 

We show reductions between the following three problems in weighted and unweighted as well as directed and undirected graphs.

\begin{definition}[$s$-$t$-shortest-path$(G,s,t)$]
	Given a graph $G$ and pointers to two verticies $s$ and $t$, return the length of the shortest path between $s$ and $t$.	
\end{definition}

\begin{definition}[Girth-Containing-Edge$(G,e)$]
	Given a graph $G$ and a pointer to an edge $e$, return the length of the shortest cycle in $G$ which contains $e$.
\end{definition}

\begin{definition}[Girth-Containing-Vertex$(G,v)$]
	Given a graph $G$ and a pointer to a vertex $v$, return the length of the shortest cycle in $G$ which contains $v$.
\end{definition}

We now begin showing that efficient reductions exist between these hard to solve linear problems. 

\begin{theorem}
	\label{thm:GirthE-to-stsp}
	Given an algorithm that solves (undirected/directed) \stsp in $f(n,|E|,M,B)$ time (undirected/directed) \GirthE can be solved in $O(f(n,|E|,M,B)+ O(1))$ time. 
\end{theorem}
\begin{proof} 
	We construct a modified graph $G'$ by taking the target edge with end vertices $v_1$ and $v_2$ and deleting it. We now run \stsp$(G',v_1,v_2)$ and return this result plus $1$. The deleted edge completes the cycle, and since it is the shortest path, it results in the smallest possible cycle. For a directed graph, ensure that the deleted edge pointed from $t$ to $s$.
\end{proof}

\begin{theorem}
	Given an algorithm that solves (undirected/directed) \GirthE in $f(n,|E|,M,B)$ time (undirected/directed) \stsp can be solved in $O(f(n,|E|,M,B)+ O(1))$ time. 
\end{theorem}
\begin{proof} 
	To obtain the shortest path between $s$ and $t$ we construct a new graph $G'$ which simply add an edge $e'$ of weight $1$ between $s$ and $t$. If this edge already exists delete it and add an edge of weight $1$. If we are in the undirected case, save the deleted edges weight as $d$. We then run \GirthE$(G',e')$ and subtract $1$. Again, the edge should be directed from $t$ to $s$ in the case of directed graphs.
	In the undirected case where we deleted an edge, compare the output shortest cycle length minus 1 to the distance $d$, return the smaller value. 
	
	 This requires a single call to \GirthE and a constant number of changes to the original input.
\end{proof}

\begin{theorem}
	Given an algorithm that solves (undirected/directed) \GirthV in $f(n,|E|,M,B)$ time (undirected/directed) \stsp can be solved in $O(f(n,|E|,M,B)+ O(1))$ time. 
\end{theorem}
\begin{proof} 
	We construct a modified graph $G'$ by adding a new vertex $v'$ and connecting it to the verticies $s$ and $t$. We then call \GirthV$(G',v')$ and return the result minus $2$. For the directed case we must direct the edges from $t$ to $v'$ and from $v'$ to $s$.
\end{proof}

\begin{theorem}
	Given an algorithm that solves (undirected/directed) \GirthV in $f(n,|E|,M,B)$ time (undirected/directed) \GirthE can be solved in $O(f(n,|E|,M,B)+ O(1))$ time. 
\end{theorem}
\begin{proof} 
	We construct a modified graph $G'$ by taking the target edge with end verticies $v_1$ and $v_2$, deleting the edge, and then replacing it with a new vertex $v'$ which connects only to $v_1$ and $v_2$. This modification only requires a constant number of operations. Next run \GirthV$(G',v')$ and return it's result minus $1$.  This reduction also works for directed graphs by directing the new edges appropriately.
\end{proof}

\begin{theorem}
	Given an algorithm that solves directed \GirthE in $f(n,|E|,M,B)$ time then directed \GirthV is solvable in $O(f(n,|E|,M,B)+n/B)$ time. 
\end{theorem} 
\begin{proof}
	We construct a modified graph $G'$ by splitting $v$ into two verticies $v_{in}$ and $v_{out}$ where $v_{in}$ contains all of the edges from other verticies to $v$ and $v_{out}$ has edges to all of the verticies which $v$ had edges to.  We then add an additional edge $e'$ directed from $v_{in}$ to $v_{out}$. We then call \GirthE$(G',e')$ and return its result minus $1$.
	
	Let $|A(v)|$ be the length of the adjacency list of the vertex $v$. These edits take time $O(|A(v)|/B)$, which is upper bounded by $O(n/B)$.
\end{proof}

\begin{theorem}
	Given an algorithm that solves directed \stsp in $f(n,|E|,M,B)$ time then directed \GirthV is solvable in $O(f(n,|E|,M,B)+n/B)$ time. 
\end{theorem} 
\begin{proof}
	We construct a modified graph $G'$ by splitting $v$ into two verticies $v_{in}$ and $v_{out}$ where $v_{in}$ contains all of the edges from other verticies to $v$ and $v_{out}$ has edges to all of the verticies which $v$ had edges to. We now run \stsp$(G',v_{out},v_{in})$ and return the result plus $1$.
	
	Let $|A(v)|$ be the length of the adjacency list of the vertex $v$. These edits take time $O(|A(v)|/B)$, which is upper bounded by $O(n/B)$.
\end{proof}

When solving \GirthV in the directed case, we know which direction the path must follow the edges and can perform this decomposition. Unfortunately this no longer works in the undirected case and a more complex algorithm is needed, giving slightly weaker results.

\begin{theorem}
	\label{thm:unGirthV-to-unstsp}
	Given an algorithm that solves undirected \stsp in $f(n,|E|,M,B)$ time then undirected \GirthV is solvable in $O((f(n,|E|,M,B)+n/B)\lg(n))$ time. 
\end{theorem} 
\begin{proof}
	When attempting to solve \GirthV in the undirected case, if we wish to split the required vertex $v$ we end up with the issue of not knowing how to partition the edges between the new nodes. However, we only need to ensure that the two edges used in the solution are assigned to opposite nodes. Conveniently, if $v$ has degree $d$ we can generate $O(\lg d)$ partitions of the edges such that every  pair of edges appears on opposite sides in at least one partition. To do so, label each edge with numbers from $0$ to $d-1$. These can be expressed by $s = \ceil{\lg(d)}$ bit numbers. We generate $s$ partitions where the the assignment of the edges in the $i^{th}$ partition is given by the value of the $i^{th}$ bit of the edge's number. Since each number is different, all pairs of them must differ in at least one bit, yielding the desired property. 
	
	To solve undirected \GirthV we first find all the neighbors of $v$ and number them as above. Now for each bit in this numbering  we construct a new graph $G_i'$ which replaces $v$ with a pair of verticies $v_{i,0}'$ and $v_{i,1}'$. Additionally, $v_{i,0}'$ is connected to all the neighbors of $v$ which had a $0$ in the $i^{th}$ bit of it's number. Similarly, $v_{i,1}'$ is connected to all the neighbors of $v$ which had a $1$ in the $i^{th}$ bit of it's number. To solve \GirthV with \stsp, after constructing each $G_i'$ we call \stsp$(G_i',v_{i,0},v_{i,1})$ and store the answer. After constructing all of the augmented graphs and running the \stsp algorithm, our girth is simply the minimum of all the shortest path lengths found. Constructing each augmented graph only requires interacting with each node and edge a constant number of times and can be done in sequential passes. It thus runs in $n/B$ time. Since this is a sparse graph, the degree of $v$ cannot be more than $|E|$ and thus we will not need to construct more than $O(\lg(n))$ graphs and make  $O(\lg(n))$ calls to the \stsp algorithm. 
\end{proof}

\begin{theorem}
	Given an algorithm that solves undirected \GirthE in $f(n,|E|,M,B)$ time then undirected \GirthV is solvable in $O((f(n,|E|,M,B)+n/B)\lg(n))$ time. 
\end{theorem} 
\begin{proof}
	The reduction from \GirthV to \GirthE proceeds exactly as in the proof of Theorem~\ref{thm:unGirthV-to-unstsp} except that we add an extra edge $e_i'$ between $v_{i,0}'$ and $v_{i,1}'$ and call \GirthE$(G_i',e_i')$ instead of \stsp on the input.
\end{proof}

\section{Lower Bounds from Fine-Grained Reductions}
The fundamental problems in the fine-grained complexity world are good starting points for assumptions in the I/O model because these problems are so well understood in the RAM model. Additionally, both APSP and 3-SUM have been studied in the I/O model \cite{apspSparse,pagh2014inputMatrix,patrascu2010towards}. These reductions allow us to propagate believed lower bounds from one problem to others, as well as propagate any potential future algorithmic improvements.  

\subsection{Reductions to 3-SUM}
\label{sec:3sum}
We will show that 3-SUM is reducible to both convolution 3-SUM and 0 triangle in the I/O-model. 

\begin{lemma}
	If convolution 3-SUM is solved in $f(n,M,B)$ time then 3-SUM  is solved in\\
	 $O(g^3 f(n/g,M,B)) + n^2/(gMB) )$ time for all $g\in [1,n]$.
	\label{lem:con3sumred}
\end{lemma}
\begin{proof}
	Following the proof of P{\v{a}}tra\c{s}cu we can hash each value into the range $n/g$ in time $n/B$ \cite{patrascu2010towards}. We then sort the elements by their hash value in time $n \lg_{M/B}(n)/B$. We scan through and put elements in over-sized buckets in one memory location, and put the elements in buckets with less than $10g$ elements elsewhere sorted by hash value in time $O(n/B)$. 
	
	The expected number of elements in an over-sized buckets is $n/g$ and then solve the 3-SUM problem on lists of length $n$, $n$ and $n/g$ in $n^2/(gMB)$ time. 
	
	We then go through the small buckets of size $<10g$ we mark each element in the buckets by their order in the bucket (so each element is assigned a unique number from $[1,10g]$). Now we re-sort the elements in small buckets by their order number in time $n \lg_{M/B}(n)$. Then we 
\end{proof}

\begin{corollary}
If convolution 3-SUM is solved in time $O(n^{2-\epsilon}/(MB))$ or $O(n^{2}/(M^{1+\epsilon}B))$  or $O(n^{2}/(MB^{1+\epsilon}))$ then 3-SUM is solved in $O(n^{2-\epsilon'}/(MB))$ or $O(n^{2}/(M^{1+\epsilon'}B))$  or $O(n^{2}/(MB^{1+\epsilon'}))$ time, violating the I/O 3-SUM conjecture. 
\label{cor:conv3to3sum}
\end{corollary}
\begin{proof}
If conv. 3-SUM (convolution 3-SUM) is solved in time $f(n,M,B)=  n^{x}/(M^yB^z)$ then we can solve 3-SUM in time $O(n^{2-(2-x)/(4-x)}/(M^{1+(1-y)/(4-x)}B^{1+(1-z)/(4-x)}))$. 

If we can solve conv. 3-SUM in $O(n^{2-\epsilon}/(MB))$ then we can solve 3-SUM in $O(n^{2-\epsilon/(2+\epsilon)}/(MB))$. 

If we can solve conv. 3-SUM in $O(n^{2}/(M^{1+\epsilon}B))$  then we can solve 3-SUM in $O(n^{2}/(M^{1+\epsilon/2}B))$.

If we can solve conv. 3-SUM in $O(n^{2}/(MB^{1+\epsilon}))$  then we can solve 3-SUM in $O(n^{2}/(MB^{1+\epsilon/2}))$. 
\end{proof}

\begin{lemma}
	If 0 triangle is solved in $f(n,M,B)$ time then convolution 3-SUM  is solved in $O(\sqrt{n}f(\sqrt{n},M,B)+n^{1.5}\lg_{M/B}(n)/B)$ time.
	\label{lem:0triangleconv3sum}
\end{lemma}
\begin{proof}
	We will use a reduction inspired by the reduction in Vassilevska-Williams and Williams \cite{williams2013finding}. We produce $\sqrt{n}$ problems. Specifically the problems will be labeled by $i\in [1,\sqrt{n}]$ and we will produce a graph on $L^i,R^i,S^i$. We make the problems as follows (as is done in Vassilevska-Williams and Williams \cite{williams2013finding}).
	\begin{align}
	w(L^i[s],R^i[t])&=A[(s-1)\sqrt{n} +t]\\
	w(R^i[s],S^i[q])&=B[(i-1)\sqrt{n} +q-t]\\
	w(L^i[s],S^i[q])&=-C[(s+i-2)\sqrt{n}+q]
	\end{align} 
	
	Zero triangle will need the input adjacency list to be given to it. Given an adjacency matrix of size $\sqrt{n}$ by $\sqrt{n}$ indexed by $k$ and $j$ let $h_{\sqrt{n}}(k,j) = k\sqrt{n}+j$. For each problem we will generate the adjacency list and lay it out in memory by labeling each element with its order in memory. It will take $n/B$ time to scan through the convolution 3-sum instance. Given the index $i$ of the problem we can compute the $k$ and $j$ (note for values in list $B$ they will have multiple pairs $k$ and $j$ produced) for the corresponding 0 triangle instance and thus compute $h_{\sqrt{n}}(k,j)$. We can scan through the values from  the lists $A$,$B$ and $C$ and assign them values $h$ and then sort the lists based on the values of $h$. This will take time $O(n/B + n\lg_{M/B}(n)/B)$ for each subproblem $i$. For a  total time of 
	$$\sqrt{n}f(\sqrt{n},M,B)+n^{1.5}\lg_{M/B}(n)/B.$$
\end{proof}

\begin{theorem}
	If 0 triangle is solved in time $O(n^{3-\epsilon}/(MB))$ or $O(n^{3}/(M^{1+\epsilon}B))$  or $O(n^{3}/(MB^{1+\epsilon}))$ then 3-SUM is solved in $O(n^{2-\epsilon'}/(MB))$ or $O(n^{2}/(M^{1+\epsilon'}B))$  or $O(n^{2}/(MB^{1+\epsilon'}))$ time, violating the I/O 3-SUM conjecture.
	\label{thm:0tri3sum}
\end{theorem}
\begin{proof}  
	If 0 Triangle is solved in time $O(n^{3-\epsilon}/(MB))$ or $O(n^{3}/(M^{1+\epsilon}B))$  or $O(n^{3}/(MB^{1+\epsilon}))$ then convolution 3-SUM is solved in $O(n^{2-\epsilon'}/(MB))$ or $O(n^{2}/(M^{1+\epsilon'}B))$  or $O(n^{2}/(MB^{1+\epsilon'}))$ time by Lemma \ref{lem:0triangleconv3sum}.
	
	By Corollary \ref{cor:conv3to3sum} we have that convolution 3-SUM being solved in  $O(n^{2-\epsilon'}/(MB))$ or\\
	 $O(n^{2}/(M^{1+\epsilon'}B))$  or $O(n^{2}/(MB^{1+\epsilon'}))$ time implies 3-SUM is solved in  $O(n^{2-\epsilon''}/(MB))$ or\\ $O(n^{2}/(M^{1+\epsilon''}B))$  or $O(n^{2}/(MB^{1+\epsilon''}))$ time. 
\end{proof}

\begin{lemma}
	If 3-SUM is solved in $f(n,M,B)$ I/Os then Convolution 3-SUM is solvable in $O(f(n,M,B) +n/B)$ I/Os.
	\label{lem:0convReduction}
\end{lemma}
\begin{proof}
	Let the largest element in the list have absolute value $M$. 
	The standard reduction takes the $i^{th}$ element of the list $A$, $A[i]=a_i$, and adds $10Mi+a_i$, $10Mi+a_i$ and $-10Mi+a_i$ to the list $L$. If $A[i]+A[j]+A[i+j]=0$ then $10Mi+A[i]+10Mj+A[j]-10M(i+j)+A[i+j] = 0$. If $10Mi+A[i]+10Mj+A[j]-10M(k)+A[k] = 0$ then due to the large size of $10M$ $i+j=k$ and $A[i]+A[j]+A[k] = A[i]+A[j]+A[i+j]=0$. 
	It takes $O(1)$ I/Os to read in $B$ entries from $A$ and write out $B$ to $L$. Thus, we run 3-SUM once on the list $L$ and take $n/B$ I/Os for a total time of $O(f(n,M,B) +n/B)$ I/Os.
\end{proof}

\begin{corollary}
	Convolution 3-SUM is solvable in $O(n^2/(MB)+n/B)$
	\label{cor:0convUB}
\end{corollary}
\begin{proof}
	Given Lemma \ref{lem:0convReduction}
\end{proof}
\subsection{APSP Reductions in the IO-Model}
\label{sec:APSPReductions}
%
%\subsection{Reductions}
%
%
%To do fine grained reductions in the caching model we can generally use the same techniques that fine grained complexity uses in the RAM model.

We show reductions between APSP, negative weight triangle finding, $(\min,+)$-matrix multiplication, and all pairs triangle detection, as diagrammed in Figure~\ref{fig:APSPandfriends}. 

%%\subsection{Lower Boudns for Caching Algorithms}

\begin{figure}[ht]
	\centering
	\includegraphics[width=0.5\textwidth]{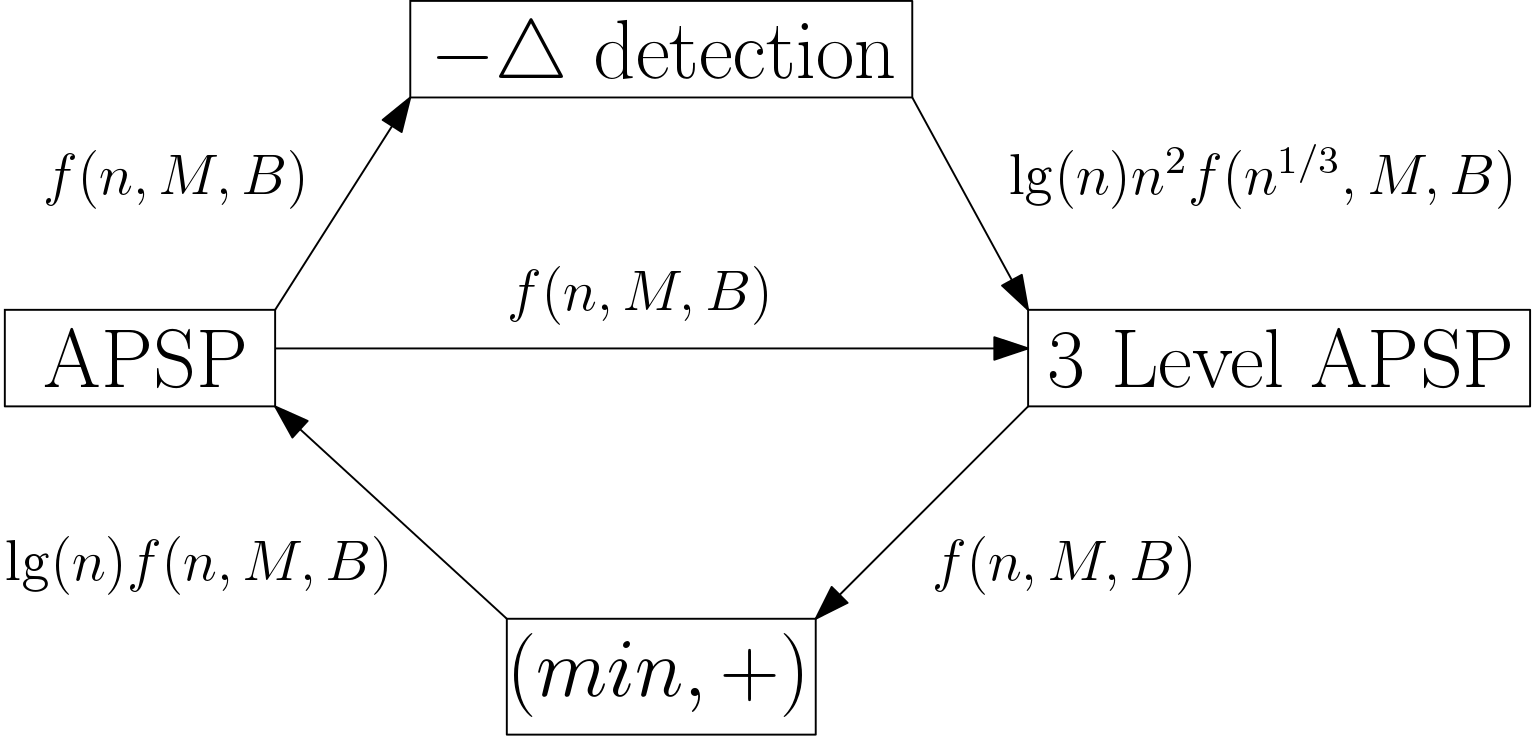}
	\caption{A depiction of the relationships between All Pairs Shortest Paths (APSP) and other problems in the I/O-model. An arrow $A \rightarrow B$ means that an algorithm for problem $A$ can be used to solve problem $B$. The labels on these arrows mean that if problem A is solvable in $f(n,M,B)$ cache misses, then problem $B$ is solvable in the time labeled on the arrow.}
	\label{fig:APSPandfriends}
\end{figure}

%We have reductions and algorithmic speedup in APSP and related cluster. Reductions (but no speed up over the trivial) in the 3SUM and related cluster. We have a potential for fast algorithms that have a reduction to OV. 

\paragraph{The problems.}

We are considering a set of four problems, shown in Figure \ref{fig:APSPandfriends}. 

\begin{definition}[All-Pairs-Shortest-Path(G)]
	Given a fully connected graph $G$ with large edge weights (weights between $-n^c$ and $n^c$ for some constant $c$) return the path lengeths between all pairs of nodes in a matrix $D$ where $D[i][j] = $ the length of the shortest path from node $i$ to node $j$). 
\end{definition}

Another related version of APSP requires us to return all the shortest paths in addition to the distances. To represent this information efficiently, one is required to return an $n$ by $n$ matrix $P$ where the $P[i][j]$ is the next node after $i$ on the shortest path from $i$ to $j$. The matrix $P$ allows one to extract the shortest path between two points by following the path through the matrix $P$.  This problem is also called APSP.

\begin{definition}[Three-Layer-APSP(G)]
 Solve APSP on $G$ where $G$ is promised to be a bipartite graph $G$ which has partitions $A$, $B$, and $C$, such that there are no edges within $A$,$B$ or $C$ and no edges between $A$ and $C$. This is shown visually in Figure \ref{fig:threelayerapsp}. 
 \end{definition}

%The first is All Pairs Shortest Paths (APSP). The problem is, \emph{ on a fully connected graph with large edge weights (weights between $-n^c$ and $n^c$ for some constant $c$) return the path lengeths between all pairs of nodes in a matrix $D$ where $D[i][j] = $ the length of the shortest path from node $i$ to node $j$)}. Another related version requires us to return all the shortest paths in addition to the distances. To represent this information efficiently, one is required to return an $n$ by $n$ matrix $P$ where the $P[i][j]$ is the next node after $i$ on the shortest path from $i$ to $j$. The matrix $P$ allows one to extract the shortest path between two points by following the path through the matrix $P$. 

\begin{definition}[Negative-triangle-detection(G)]
	Given a graph $G$, retuning true if there is a negative triangle and false if there is no negative triangle. This problem is also called $-\triangle$ detection.
\end{definition}

%Negative triangle detection ($-\triangle$ detection) is the problem of,\emph{ given a graph, retuning true if there is a negative triangle and false if there is no negative triangle}.

\begin{figure}[ht]
	\centering
	\includegraphics[width=0.5\textwidth]{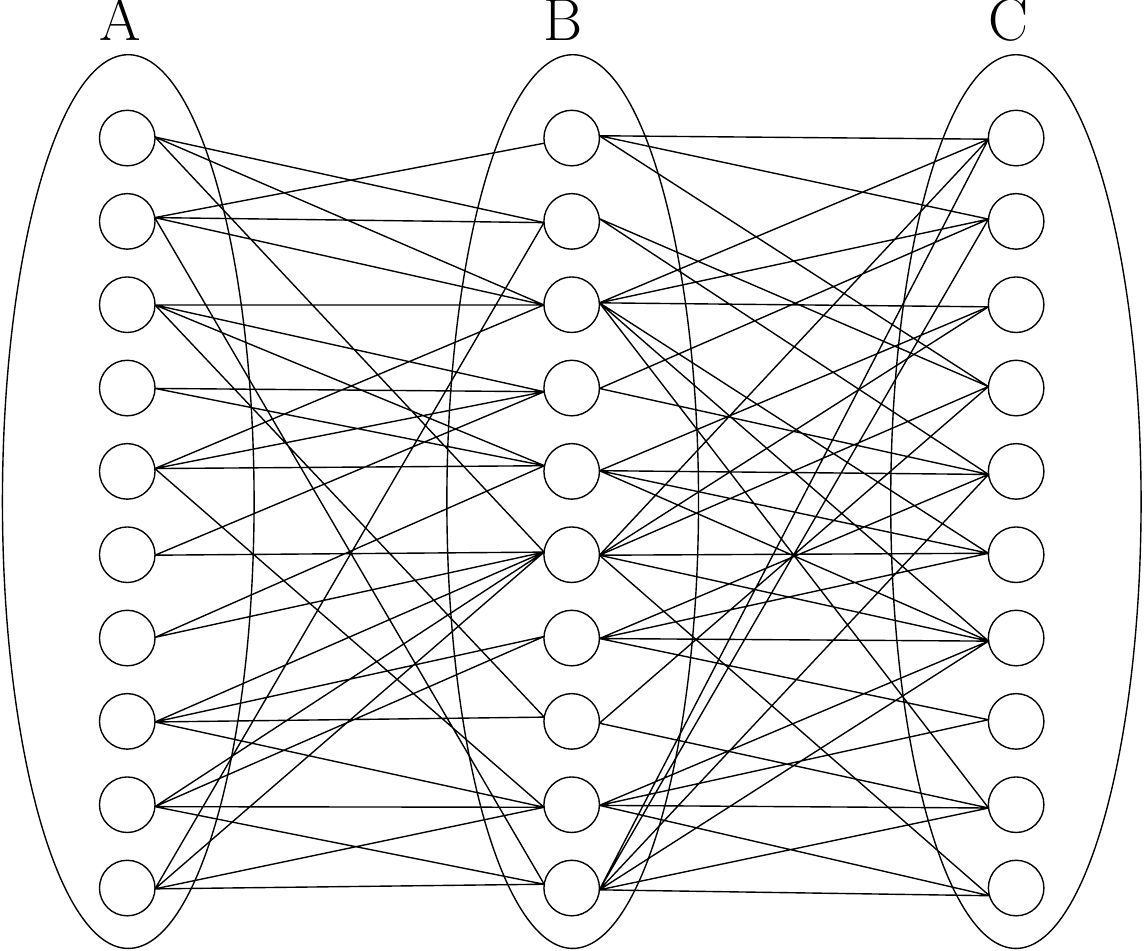}
	\caption{An example of a three layer APSP problem.}
	\label{fig:threelayerapsp}
\end{figure}

\begin{definition}[$(\min,+)$-Matrix-Multiplication(A,B)] 
This problem is a variant on matrix multiplication.Given an $n$ by $n$ matrix $A$ and an $n$ by $n$ matrix $B$ return an $n$ by $n$ matrix $C$ such that $C[i][j]= \min \left( \{A[i][k]+B[k][j]| \forall k\in [1,n]  \} \right )$ .
\end{definition}

The motivation for showing I/O equivalences between these problems is two fold. First, just as in the RAM model, these reductions can provide a shared explanation for why some problems have seen no improvement in their I/O complexity for years. 

\paragraph{The set of reductions.}

\begin{theorem}
	If $\left(\min,+\right)$ runs in time $f\left(n,M,B\right)$, then APSP runs in 
	$O \left ( \lg\left(n\right) f\left(n,M,B\right) \right )$. 
	\label{thm:minTOapsp}
\end{theorem}
\begin{proof}
	In folklore we can solve APSP with $lg\left(n\right)$ calls to $\left(\min,+\right)$. We take the adjacency matrix $A$ and then use repeated squaring to produce $A, A^2, A^{2^2}, \ldots , A^{2^{\lg\left(n\right)}}$. Then we simply multiply these $\lg\left(n\right)$ matrices together and the output will be the set of shortest paths between all pairs.

	To get both the min path lengths and the successor matrix after each multiplication, we will use both $V$ and $S$ output by $\left(\min,+\right)$. Say we are multiplying $L_1$ and $L_2$ and they have successor matrices $S_1$ and $S_2$, and the output of the $\left(\min,+\right)$ multiplication is $\left(V,S\right)$. The output length matrix $L_{out} = V$ and $S_{out}= S$. 
	
	Each multiplication takes $O\left(n^2/B\right)$ cache misses to read in and write out the matrices and $f\left(n,M,B\right)$ for the multiplication itself. Note that a trivial lower bound on $f\left(n,M,B\right)$ is $O\left(n^2/B\right)$.
	So the total number of cache misses is $O \left ( \lg\left(n\right) f\left(n,M,B\right)  \right )$
\end{proof}

\begin{theorem}
	If All Pairs Shortest Paths runs in time $f\left(n,M,B\right)$, then negative weight triangle detection in a tripartite graph runs in 
	$O \left ( f\left(n,M,B\right) \right )$ 
	cache misses. 
\end{theorem}
\begin{proof}
	Let us call the  whole vertex set $V$ and the three groups of nodes $I$, $J$ and $K$.\\
	Let $m = \max \{ |w\left(v,u\right)| | \forall v,u \in V \}$. 
	
	 Now we create a new graph $G'$. Where $V= I \cup J \cup K \cup I'$ and all edges $\left(i,k\right)$ for $i\in I$ and $k\in K$ are removed and the edge $\left(i',k\right)$ is added. Additionally we add $7m$ to the weights of all edges (this to force the shortest paths to not 'backtrack' and go through one set multiple times). This takes $O\left(n^2/B\right)$ cache misses. 
	 
	 Now we run $APSP$ on $G'$. We look at the path lengths between pairs of nodes $i$ and $i'$. If any of those path lengths is $<21m$, then the total original triangle was negative, return true. Otherwise we return false. 
	 
	 Setting up the graph takes $O\left(n^2/B\right)$ cache misses. Running APSP takes $f\left(4/3n,M,B\right)$ cache misses. Checking for short paths between $i$ and $i'$ takes $O\left(n\right)$ time. If $n>B$, then $n = O\left(n^2/B\right)$. If $n<B$, then the entire computation fits in two cache lines and thus takes $O\left(1\right )$ time to compute even if $M = \Theta\left(B\right)$. Once again $n^2/B$ is a trivial lower bound on $f\left(n,M,B\right)$. So the total number of cache misses is
	 $$O \left ( f\left(n,M,B\right)  \right ).$$
	
\end{proof}

\begin{theorem}
	If  negative weight triangle detection in a tripartite graph runs in time $f\left(n,M,B\right)$, then three layer APSP with weights in the range $[-W,W]$ runs in 
	$O \left (\lg\left(W\right)n^2 f\left(n^{1/3},M,B\right) \right )$ 
	cache misses.
	\label{thm:negtriangleTOmintriangle} 
\end{theorem}
\begin{proof}
	
	We will use the same reduction as \cite{williams2010subcubic} and analyze it in I/O-model. 
	
	Let the three layer APSP's layers be called $I$, $J$ and $K$. We want to find for every pair $\left(i,k\right)$ where $i\in I$ and $k\in K$ the $j$ such that triangle $\triangle_{i,j,k}$ has minimum weight. 
	We will discover this by doing $\lg\left(W\right)+1$ rounds where we start by re-assigning all $w\left(i,k\right)=0$ and then binary search on each $w\left(i,j\right)$ for the value where $w\left(i,j\right)+w\left(j,k\right)+w\left(k,i\right)=0$.
	
	Now, in each round we split each set $I$, $J$ and $K$ into $n^{2/3}$ groups of size $n^{1/3}$.  We can then once again keep two matrices $V$ for minimum value so far and $S$ for the $j$ achieving that value. 
	
	We can call negative weight triangle detection repeatedly on all $\left(n^{2/3}\right)^3$ possible choices of three subsets. This will take at most $n^2+\left(n^{2/3}\right)^3$ calls. One for each edge removed and one for each subset. This results in $O\left(n^2f\left(n^{1/3},M,B\right)\right)$ cache misses.
	The total number of cache misses is 
	$$ O \left ( \lg\left(W \right)n^2 f\left(n^{1/3},M,B\right) \right ) .$$
\end{proof}

\begin{corollary}
	If  negative weight triangle detection in a tripartite graph runs in time $f\left(n,M,B\right)$, then three layer APSP over weights in the range $[-poly\left(n\right),poly\left(n\right)]$ runs in 
	$O \left (\lg\left(n\right)n^2 f\left(n^{1/3},M,B\right)\right)$ 
	cache misses. 
\end{corollary}
\begin{proof}
	Simply apply Theorem \ref{thm:negtriangleTOmintriangle} with a $poly\left(n\right)$ weight. 
\end{proof}

\begin{theorem}
	If three layer APSP runs in time $f\left(n,M,B\right)$, then $\left(\min,+\right)$  matrix multiplication runs  in 
	$O \left (f\left(n,M,B\right) \right )$ 
	cache misses. 
	\label{thm:apspTOnegtri}
\end{theorem}
\begin{proof}
	Given an instance of $\left(\min,+\right)$ matrix multiplication produce a graph $G$ made up of three sets of size $n$: $I$,$J$ and $K$. Edges will go from $I$ to $J$ and $J$ to $K$.
	
	The length of the edge from $i \in I$ to $j \in J$  will be $w\left(i,j\right) = A[i,j]$. 
	The length of the edge from $j \in J$ to $k \in K$  will be $w\left(j,k\right) = B[j,k]$. 
	The length of the edge from $k \in K$ to $i \in I$ will be $w\left(k,i\right) = 0$.
	
	Run ``All pairs min triangle detection in a tripartite graph'' on $G$ and it produces a matrix that lists the $j$ that minimize the triangles $S$. Return a matrix $S$ of $j$ and a matrix $V$ where $V[i,k] = w\left(i,S[i,j]\right)+w\left(S[i,j],j\right)$. Return $V$ and $S$. 
	
	This takes $O\left(n^2/B + f\left(n,M,B\right)\right)$ cache misses 
\end{proof}

\begin{lemma}
	If All Pairs Shortest Paths runs in time $f\left(n,M,B\right)$, then three layer APSP runs in 
	$O \left ( f\left(n,M,B\right) \right )$ 
	cache misses. 
	\label{thm:apspTOmintriang}
\end{lemma}
\begin{proof}
 	Three layer APSP is APSP but with the possible inputs reduced. Running APSP will solve the three layer promise problem because APSP works on any graph. 
\end{proof}

\paragraph{The equivalences.}

\begin{theorem}
	The following problems run in time $\widetilde{O}(\frac{N^3}{\sqrt{M}B}+ \frac{n^2}{B})$
	 cache misses:
	 \begin{enumerate}
	 	\item All Pairs Shortest Paths
	 	\item $\left(\min,+\right)$ matrix multiplication
	 	\item Negative triangle detection in a tripartite graph 
	 	\item Three layer APSP
	 \end{enumerate}
	 \label{thm:fastAlgorithms}
\end{theorem}
\begin{proof}
	This was proven true for $\left(\min,+\right)$ matrix multiplication and APSP by previous work~\cite{jia1981complexity}.
	The reduction in Theorems \ref{thm:apspTOmintriang} and \ref{thm:apspTOnegtri} we get that both  Negative triangle detection in a tripartite graph 
	and three layer APSP run in $\widetilde{O}\left(\frac{N^3}{\sqrt{M}B}+ \frac{n^2}{B}\right)$.
\end{proof}

\begin{corollary}
	The following solve APSP faster.
	\begin{enumerate}
		\item If $\left(\min,+\right)$ matrix multiplication is solvable in $f(n)$ time then APSP is solvable in $O(lg(n) f(n,M,B))$ time. 
		\item If negative triangle detection in a tripartite graph is solvable in $f(n)$ time then APSP is solvable in $O(lg^(n)n^2f(n^{1/3},M,B))$ time.
	\end{enumerate}
\end{corollary}
\begin{proof}
	By using the
	reductions in Theorems \ref{thm:apspTOmintriang} and \ref{thm:apspTOnegtri} we get these values. 
\end{proof}

\begin{lemma}
If 0 triangle is solved in $f(n,M,B)$ time then $- \triangle$ over numbers in the range of $[-W,W]$ is solved in $O(lg(W)f(n,M,B)+ lg(W)n^2/B)$.
\label{lem:zeroTraingleAPSP}
\end{lemma}
\begin{proof}
	Following the reduction from Vassilevska-Williams  and Williams we will turn negative triangle on the graph $G$ into $\lg(W)$ copies of the problem \cite{williams2010subcubic}. We will create a tripartite instance of the $-\triangle$ problem by making $3$ copies of the vertex set $V',V'',V'''$ and $e(v',w'')=e(v',w''')=e(v'',w''')=e(v,w)$ but $e(v',w')=e(v'',w'')=e(v''',w''')=\infty$. We then consider the $\lg(W)$ problems created by replacing edge weights with the highest $i$ bits of that edge length. Creating these new problems takes $n^2/B$ time, there are $\lg(W)$ problems we need to write. So we take total time $O(lg(W)f(n,M,B)+ lg(W)n^2/B)$.
\end{proof}

\begin{lemma}
	Zero triangle is solvable in $O(n^3/(\sqrt{M}B)+n^2/B)$ I/Os. 
	\label{lem:0triUB}
\end{lemma}
\begin{proof}
We can consider the tripartite version. There are three sets of vertices $|A|=|B|=|C|=n$. Let $A_i= A[in/g,(i+1)n/g]$, $B_i= B[in/g,(i+1)n/g]$ and $C_i= C[in/g,(i+1)n/g]$. Then, we can consider the $g^3$ subproblems $A_i$, $B_j$ and $C_k$ where $i,j,k \in [1,g]$. Every triangle is contained in some subproblem. We can fit a subproblem in memory if $|A_i|=\sqrt{M}$. This gives us I/Os $n^3/M^{1.5} (MB) + n^2/B= O(n^3/(\sqrt{M}B)+n^2/B)$. 
\end{proof}

%\begin{lemma}
%Say we can solve $A$ in $f_A\left(n,M,B\right)$ cache misses. If $B$ can solved with $t\left(n,M,B\right)$ cache misses and $c$ calls to $A$ on input sizes $n_1,n_2, \ldots , n_c$ then $B$ can solved with 
%$ O\left(t\left(n,M,B\right) + \sum_{i=0}{c} f_A\left(n_i,M,B\right))$ cache misses.
%\end{lemma}
%\begin{proof}
%\end{proof}

%\subsubsection{3-SUM and related problems}

%\begin{theorem}
%	If $3SUM$ is solvable in $f\left(n,M,B\right)$ cache misses then GEOMBASE is solvable in $O\left(f\left(n,M,B\right) +\frac{n}{B}\right)$.
%\end{theorem}

%\begin{theorem}
%	If GEOMBASE is solvable in $f\left(n,M,B\right)$ cache misses then $3SUM$ is solvable in $O\left(f\left(n,M,B\right) +\frac{n}{B}\right)$.
%\end{theorem}

%\subsubsection{Orthogonal Vectors and related problems}

%\xxx{Reductions to OV? Many from OV, but this non-trivial $MB$ improvement does not much help in that direction :P}

%\begin{theorem}
%	Orthogonal Vectors is solvable in  
%	$O(\frac{n^2d^2}{MB} + \frac{n}{B})$ cache misses.
%\end{theorem}
%\begin{proof}
%\end{proof}

%Note this means we can do almost as much as you would hope for with our cache, the best we could hope for would be $2^{-M}$.

\subsection{Orthogonal Vectors (OV)}
\label{sec:ov}

\begin{lemma}
	OV is solvable in $O(n^2/(MB) + n/B)$ I/Os cache obliviously.
	\label{lem:ovUB}
\end{lemma}
\begin{proof}
	We will give a recursive algorithm, $OV(A,B)$. 
	
	The base case is when $|A|=|B|=1$, then we simply take the dot product. If the dot product is zero then return $TRUE$. 
	
	Given two lists $A$ and $B$ of size greater than one then divide the two lists in half. Call the halfs of $A$ $A_1$ and $A_2$. Call the halves of $B$ $B_1$ and $B_2$. We then run $OV$ on four recursive calls $OV(A_1, B_1)$, $OV(A_1, B_2)$, $OV(A_2, B_1)$, and $OV(A_2, B_2)$. If any return $TRUE$ then return true, else return $FALSE$.
	
	The running time of this algorithm is given by 
	$T(n)= 4T(n/2)+n/B$ 
	
	We use the master's theorem from Section \ref{sec:masterstheorem} to find that the running time is $O(n^2/(MB) + n/B)$. We note that access pattern of this algorithm is independent of the size of cache and the size of the cache line. Making this algorithm cache oblivious. 
\end{proof}

\begin{lemma}
	If sparse diameter is solvable in $f(n,M,B)$ I/Os then OV is solvable in $\tO( f(n,M,B)+ n/B)$ I/Os.
	\label{lem:diamOV}
\end{lemma}
\begin{proof}
	We will use the reduction from Abboud, Vassilevska-Williams and Wang \cite{diameterReduction}. 
	If the lists from OV are $A$ and $B$ with $n$ vectors of length $d$. We have $n$ nodes $a_i$, $n$ nodes $b_i$, $d$ nodes $v_i$ and two nodes $y_a$ and $y_b$. Where the edges $(a_i,y_a)$ all exist and edges $(d_i,y_a)$ all exist. Where the edges $(b_i,y_b)$ all exist and edges $(d_i,y_b)$ all exist. The edges $(a_i,d_j)$ exist if $A[i][j]=1$ and the edges  $(b_i,d_j)$ exist if $B[i][j]=1$. 
	
	We can in time $n/B$ output an adjacency list for $y_a$ and $y_b$.
	We read in as many vectors (or fractions of a vector as we can) and we output  $(a_i,d_j)$ and  $(d_j,a_i)$ if $A[i][j]$ and $(b_i,d_j)$ and  $(d_j,b_i)$ if $B[i][j]$. We then sort these vectors which takes $\frac{nd}{B} \lg_{\frac{M}{B}}{nd}$ time. This produces adjacency lists for all nodes. 
\end{proof}

\begin{lemma}
	If Edit Distance is solvable in $f(n,M,B)$ time then $OV$ is solvable in $\tO(f(n,M,B)+n/B)$ I/Os. 
	\label{lem:edOV}
\end{lemma}
\begin{proof}
	The edit distance proof from Backurs and Indyk is generated by taking each vector and making a string not too much longer \cite{backurs2015edit}. The total time to produce the strings is $\tO(n/B)$.
	
	There are two stages of the reduction. 
	
	The first stage reduces orthogonal vectors to a problem they define, Pattern. In this reduction the vectors from the edit distance problem are converted into a string that can be formed by reading the vectors bit by bit. This means the pattern can be produced in $\tO(n/B)$ time.
	
	The second step of the reduction reduces pattern to edit distance. This works by tripling the size of one of the patterns. This also takes $\tO(n/B)$ time. 
	
	As a result, we can take the original orthogonal vector problem and turn it into two strings that can be given as input to edit distance. Thus, if the edit distance problem is solvable in time $f(n,M,B)$ then orthogonal vectors is solvable in time $\tO(f(n,M,B)+n/B)$
\end{proof}

\begin{lemma}
	If Longest Common Subsequence is solvable in $f(n,M,B)$ time then $OV$ is solvable in $\tO(f(n,M,B)+n/B)$ I/Os.
	\label{lem:LCSOV}
\end{lemma}
\begin{proof}
	We will use the reduction from Abboud, Backurs, and  Vassilevska-Williams \cite{abboud2015tight}. Their reduction produces strings by going through the vectors bit by bit and generating two strings of length $O(nd^2)$. The LCS of the two strings tells us if the original orthogonal vectors instance had an orthogonal pair of vectors, if there are $O(n)$ strings each of length $O(d)$ in the orthogonal vectors instance. The reduction can be performed in the I/O-model in $nd^2/B$ I/Os.
	
	Thus if LCS is solvable in $f(n,M,B)$ time then orthogonal vectors is solvable in time\\ 
	$O(f(nd^2,M,B)+nd^2/B)$ time which takes $\tO(f(n,M,B)+n/B)$ I/Os.
\end{proof}

\subsection{Hitting Set}
Abboud, Vassilevska-Williams and Wang define a new problem, called Hitting Set~\cite{diameterReduction}. 

\begin{definition}[Hitting Set]
Given an input of a list $A$ and $B$ of $d$ dimensional $1$ and $0$ vectors return $True$ if there $\exists a\in A$ such that $\forall b\in B$ $a \cdot b >0$.
\end{definition}

\begin{lemma}
	Hitting Set is solvable in $\tO( n^2/(MB)+ n/B)$ I/Os.
	\label{lem:hsUB}
\end{lemma}
\begin{proof}
	We can solve hitting set by splitting our input lists into groups $A_i = A[iM/2,(i+1)M/2]$ and  $B_i = B[iM/2,(i+1)M/2]$ we can pull in every pair of $A_i$ and $B_j$ each pair takes $M/B$ I/Os to read in and mark an index $A_i[j]$ if we find that there is a vector that $A_i[j]$ is orthogonal to. There are $n^2/M^2$ pairs $A_i$ and $B_j$. This gives us a total number of I/Os of $n^2/(MB)$.
\end{proof}

\begin{lemma}
	If sparse radius is solvable in $f(n,M,B)$ I/Os then Hitting Set is solvable in
	\\ $\tO( f(n,M,B)+ n/B)$ I/Os.
	\label{lem:radiHS}
\end{lemma}
\begin{proof}
	We will generate the same graph as in our sparse diameter reduction. 	We will use the reduction from Abboud, Vassilevska-Williams and Wang \cite{diameterReduction}. 
	If the lists from OV are $A$ and $B$ with $n$ vectors of length $d$. We have $n$ nodes $a_i$, $n$ nodes $b_i$, $d$ nodes $v_i$ and two nodes $y_a$ and $y_b$. Where the edges $(a_i,y_a)$ all exist and edges $(d_i,y_a)$ all exist. Where the edges $(b_i,y_b)$ all exist and edges $(d_i,y_b)$ all exist. The edges $(a_i,d_j)$ exist if $A[i][j]=1$ and the edges  $(b_i,d_j)$ exist if $B[i][j]=1$. 
	
	We can in $n/B$ time output an adjacency list for $y_a$ and $y_b$.
	We read in as many vectors (or fractions of a vector as we can) and we output  $(a_i,d_j)$ and  $(d_j,a_i)$ if $A[i][j]$ and $(b_i,d_j)$ and  $(d_j,b_i)$ if $B[i][j]$. We then sort these vectors which takes $nd \lg_{M/B}{nd}/B$ time. This produces adjacency lists for all nodes. 
\end{proof}

\section{I/O Model Complexity Classes}
\label{sec:Pcache}
In this section we examine the I/O model from a complexity theoretic perspective. Section~\ref{sec:OracleModel} provides some necessary background information. In Section~\ref{sec:PandPSPACE} we define the classes $\PCACHE_{M,B}$ and $\CACHE_{M,B}\left(t\left(n\right)\right)$ describing the problems solvable in a polynomial number of cache misses and $O\left(t\left(n\right)\right)$ cache misses respectively. We then demonstrate that $\PCACHE_{M,B}$ lies between $P$ and $PSPACE$ for reasonable choices of cache size. In Section~\ref{sec:TMsimRAM} we provide simulations between the I/O model and both the RAM and Turing machine models. In Section~\ref{sec:Hierarchy} we prove the existence of a time hierarchy in $\CACHE_{M,B}\left(t\left(n\right)\right)$.
The existence of a time hierarchy in the I/O model grounds the study of fine-grained complexity by showing that such increases in running time do provably allow more problems to be solved. The techniques to achieve many of the results in this section also follow in the same theme of reductions, although the focus of the problems examined is quite different.
\subsection{Hierarchy Preliminaries: Oracle Model}
\label{sec:OracleModel}

Oracles are used to prove several results, most notably in the time-hierarchy proof. The oracle model was introduced by Turing in 1938 \cite{turing1938systems}. The definition we use here comes from Soare \cite{soare1999recursively} and similar definitions can be found in computational complexity textbooks.

In the Turing machine oracle model of computation we add a second tape and corresponding tape head. This oracle tape and its oracle tape head can do everything the original tape can:  reading, writing and moving left and right. This oracle tape head has two additional states $ASK$ and $RESPONSE$. After writing to the oracle tape, the tape head can go into the $ASK$ state. In the $ASK$ state the oracle computation is done on the input written to the oracle tape and then the tape head is changed to the $RESPONSE$ state. All of this is done in one computational step. If the oracle is the function language $L:~\{0,1\}^n \rightarrow \{0,1\}^*$, then the output is written on the tape for the input $i$ is $L(i)$. This allows the Turing machine to make $O(1)$ cost black box calls to the oracle language and get strings as output. 

%\begin{figure}[ht]
%	\centering
%	\includegraphics[width=0.4\textwidth]{TemporaryOracleDiagram}
%	\includegraphics[width=0.4\textwidth]{TemporaryOracleRAMDiagram}
%	\caption{The Turing machine oracle model and the RAM oracle model. }
%	\label{fig:oracle}
%\end{figure}

The notation $A^B$ describes a computational class of the languages decidable by an oracle Turing machine of $A$ with a $B$ oracle.  The oracle language will be a language decidable in the function version of $B$. The oracle machine will then be resource limited as $A$ is resource limited. 

In this paper we will also talk about RAM machine oracles. This is a simple extension of the typical Turing machine oracle setup. The RAM machine will have two randomly accessible memories. One will be the standard RAM memory. The other memory will be the oracle memory, the RAM can read and write words to this memory and can additionally enter the $ASK$ state. One time step after entering the $ASK$ state the RAM will be returned to the $RESPONSE$ state and the contents of the oracle memory will contain the oracle language output, $L(i)$.  

\subsection{$\PCACHE_{M,B}$ and its relationship with P and PSPACE }
\label{sec:PandPSPACE}
First we define the class of problems solvable 
 given some function, $t(n)$, the number of cache misses, up to constant factors. 
\begin{definition}
	Let $\CACHE_{M,B}\left(t\left(n\right)\right)$ be the set of problems solvable in $O\left(t\left(n\right)\right)$ cache misses on a IO-model machine with a cache of size $O(M)$ and a cache line size of size $O(B)$.
\end{definition}

Next, we consider the class of problems solvable by any polynomial number of cache misses. 
\begin{definition}
	Let $\PCACHE_{M,B}$ be the set of problems solvable in polynomial numbers of cache misses on a IO-model machine with a cache of size $O(M)$ and a cache lines of size $O(B)$.
	$$\PCACHE_{M,B} = \bigcup_{i=_1}^\infty CACHE_{M,B}(n^i)$$
\end{definition}

First, let us note that the CACHE class can simulate the RAM class. The IO-model is basically a RAM model with extra power. 

\begin{lemma}
\label{lem:RAMsubsetCache}
$RAM(t(n)) \subseteq CACHE_{M,B}(t(n))$
\end{lemma}
\begin{proof}
With $B=1$ and $M=3$ we can simulate all the operations of a RAM machine. 
There are three things to simulate in $O(1)$ cache misses: 
\begin{itemize}
	\item Write a constant to a given word in memory. \\
	We can write a word to memory with 1 cache miss. 
	\item Read $a$ and write $op(a)$ to a given word in memory.\\
	We can read in $a$ with one cache miss. We can compute $op(a)$ with zero cache misses for the operations doable in one time-step on the RAM model. Finally, we can write $op(a)$ with one cache miss. For a total of two cache misses. 
	\item Read $a$ and $b$ and write $op(a,b)$ to a given word in memory. \\
	 We can read $a$ and $b$ with two cache misses. We can compute $op(a,b)$ in zero cache misses for all operations that are doable in one time step for the RAM model. Finally, we can write $op(a,b)$ in one cache miss. 
\end{itemize}
$\CACHE_{M,B}(t(n))$ defines $M$ and $B$ asymptotically and thus is a superset of $\CACHE$ with any constant $M$ and $B$.
\end{proof}

Now we introduce a complexity class $\MEM$. Note that this class is very similar to $SPACE$. 

\begin{definition}
	We define the class $\MEM(s(n))$ to be the set of problems solvable in $SPACE(s(n))$ when the input is of size $O(s(n))$.
\end{definition}

Why $\MEM$ and not $SPACE$? We want to use the $\MEM$ class as an oracle which will model computation doable on a cache machine in one cache miss. 
When $t(n) = \Omega(n)$ then $\MEM(t(n)) = SPACE(t(n))$; however, these classes differ when we have a small work space. 
A $SPACE( o(n) )$  machine is given a read-only tape of size $n$ and compute space $o(n)$. This extra read-only tape gives the $SPACE$ machine too much power when compared with the cache. Notably, we can scan through the entire input with one step with a $SPACE(\lg(n))$ oracle. A cache would require $n/B$ time to scan this input.

%A $SPACE(Mw)$ machine is almost the right class for this. We can write information into the oracle tape and solve problems in $SPACE(Mw)$ on a cache. However, the $SPACE$ class can have access to a write-only tape with a large input as long as the read-only tape is of size $Mw$. Note that taking a small memory and repeatedly scanning through a large input would cause many cache misses. The additional restriction on the class $\MEM$ limits us to cases where the space allocated for computation is the same as the space allocated for input. 

\begin{lemma}
	$\MEM(Mw) \subseteq CACHE_{M,B}(M/B)$.
\end{lemma}
\begin{proof}
We can read the entire problem of size $Mw$ into memory by bringing in $Bw$ bits at a time, for a total of $M/B$ cache misses. 
Once the problem is in memory, it can be solved entirely in cache with a cache of size $O(Mw)$ bits, or $O(M)$ words.  
\end{proof}

Now we prove that a RAM machine with oracle access to our $\MEM$ oracle can be simulated by a cache machine. We simulate the MEM machine and RAM machine together efficiently in cache. 

\begin{corollary}
\label{cor:ram-subset-cache}
	$RAM(t(n))^{MEM(Mw)} \subseteq CACHE_{M+3,B}(t(n))$.
\end{corollary}
\begin{proof}
	We can use $3$ words in the cache to simulate our $RAM(t(n))$ machine using Lemma \ref{lem:RAMsubsetCache}. We can use the remaining $M$ words of the cache to simulate the $\MEM(Mw)$ oracle. Each word the RAM machine writes to the oracle tape can be simulated in $O(1)$ cache misses (pull from main memory and write to the simulated oracle tape). Each word the RAM machine reads from the oracle can be simulated with no cache misses, because both the oracle tape and the RAM simulation are in cache.
\end{proof}

The class $\PCACHE_{M,B}$ is equivalent to a polynomial time algorithm with oracle access to a MEM oracle. Intuitively, in both cases we get to use a similarly powerful object (the cache or the MEM oracle) a polynomial number of times. 

\begin{theorem}
	If $Bw= O\left(poly\left(n\right)\right)$, then $\PCACHE_{M,B} = P^{MEM\left(Mw\right)}$ 
	\label{thm:cacheSpaceEquivalence}
\end{theorem}
\begin{proof}
	
First, we consider the inclusion $\PCACHE_{M,B} \subseteq P^{MEM(Mw)}$. %The oracle will be the language that takes in a sate of a cache and returns the state of that cache at the next cache miss. That is, the oracle is simulating the "free" RAM computation we allow the cache in between two cache misses.  
 The oracle tape will have the memory address of the next requested cache line, a space for the RAM machine to write the contents of the requested cache line, the state the RAM simulation ended in, and finally the $Mw$ bits of the contents of the cache.  The oracle will be queried once per simulated cache miss, and return the state of the cache at the next cache miss, as well as the requested cache line. The polynomial machine will write the requested cache line to the oracle tape and then run the oracle. 
Note, if the cache machine we simulate takes $t(n)$ cache misses the polynomial machine (of the $P^{MEM(Mw)}$ oracle machine), then will need to take time $O(Bwt(n))$. However, $Bw$ is $O(poly(n))$ and $t(n)$ is $O(poly(n))$, so $Bwt(n)$ is also $O(poly(n))$. 

Second,  we consider the inclusion $\PCACHE_{M,B}=\PCACHE_{M+3,B} \supseteq P^{MEM(Mw)}$. Note that  $\PCACHE_{M+3,B}$ and $\PCACHE_{M,B}$ are equal because we consider the asymptotic size of the cache and cache line. The first $3$ words of the cache will be used to simulate $P$ (by simulating a RAM machine). The next $M$ words will be used to simulate the $\MEM(Mw)$ oracle and its tape. We can simulate the $\MEM(Mw)$ oracle with our cache because we can run any RAM program that uses only $Mw$ space on the cache in $1$ time step. 

Given $\PCACHE_{M,B} \subseteq P^{MEM(Mw)}$ and $\PCACHE_{M,B} \supseteq P^{MEM(Mw)}$ we have that $\PCACHE_{M,B} = P^{MEM(Mw)}$.
\end{proof}

We then note that in many cases MEM and SPACE are equivalent. 

\begin{corollary}
	If $Bw= O\left(poly\left(n\right)\right)$ and
	$Mw= \Omega\left(n\right)$, then \\
	 $\PCACHE_{M,B} = P^{SPACE(Mw)} $.
\end{corollary}
\begin{proof}
	Given Theorem \ref{thm:cacheSpaceEquivalence}, $\PCACHE_{M,B} = P^{MEM\left(Mw\right)}$ . We further have that $\MEM(s(n)) = SPACE(s(n))$ if $s(n) = \Omega(n)$. Note that the definition of $\MEM(s(n))$ is problems solvable in $s(n)$ space with an input of size $O(s(n))$. SPACE machines are given an input in their working tape (and thus an input of size $O(s(n))$) when $s(n) =\Omega(n)$.   
\end{proof}

Finally, we note that $P$ is a subset of $\PCACHE_{M,B}$.

\begin{corollary}
	
	$P \subseteq \PCACHE_{M,B} $.
	\label{lem:PsubPCACHE}
\end{corollary}
\begin{proof}
	By $P \subseteq P^{MEM(Mw)}$ and by Theorem \ref{thm:cacheSpaceEquivalence}  $P \subseteq P^{MEM(Mw)} = \PCACHE_{M,B}$.
\end{proof}

\begin{lemma}
	If $Mw= \Theta\left(poly\left(n\right)\right)$, then $\PCACHE_{M,B} \subseteq PSPACE$ 
	\label{thm:PCACHEPSAPCE}
\end{lemma}
\begin{proof}
	
First, we consider the inclusion $\PCACHE_{M,B} \subseteq PSPACE$. %The oracle will be the language that takes in a sate of a cache and returns the state of that cache at the next cache miss. That is, the oracle is simulating the "free" RAM computation we allow the cache in between two cache misses.  
The $\PCACHE$ can only pull in $O\left(wB poly(n)\right)$ bits from main memory which is polynomial since $wB \leq M = O(poly(n))$. Our PSPACE machine reserves two polynomial size sections of tape, one to simulate the cache and the other to store all of the values the $\PCACHE$ machine pulls from main memory. Thus $\PCACHE_{M,B} \subseteq PSPACE$.

%Second,  we consider the inclusion $\PCACHE_{M,B} \supseteq PSPACE$. The PCACHE machine has a polynomial sized cache. We use a linear number of cache misses to pull our input into cache and then use the polynomially sized cache to simulate the entire PSPACE computation since the PSPACE machine is bounded to never use more than a polynomial amount of space.s

%Given $\PCACHE_{M,B} \subseteq PSPACE$ and $\PCACHE_{M,B} \supseteq PSPACE$ we have that $\PCACHE_{M,B} = PSPACE$.
\end{proof}

\begin{lemma}
	$ \bigcup_{c=_1}^\infty PCACHE_{n^c,B} = PSPACE$ 
	\label{thm:sumPCACHEeqPSPACE}
\end{lemma}
\begin{proof}
First, we show any language in $PSPACE$ is in $\bigcup_{c=_1}^\infty PCACHE_{n^c,B}$. Every language, $L$, in $PSPACE$ is computable in $SPACE(n^{t_L})$ for some constant $t_L$. Next note that $\PCACHE_{n^{\lceil t_L \rceil},B} \supset SPACE(n^{t_L})$, because $n^{\lceil t_L \rceil} = \Omega(n)$ thus $\MEM(n^{\lceil t_L \rceil})= SPACE(n^{\lceil t_L \rceil})$. So, every language $L$ in $PSPACE$ is contained in $\bigcup_{c=_1}^\infty PCACHE_{n^c,B}$. 

Second, every language, $L$, in 	$\bigcup_{c=_1}^\infty PCACHE_{n^c,B}$ is contained in  $\CACHE_{n^{x_L},B}(n^{y_L})$ for some constants $x_L$ and $y_L$. Note that $\CACHE_{n^{x_L},B}(n^{y_L}) \subset SPACE(n^{x_L}+n^{y_L})$. The sum of two polynomials is a polynomial, so any language in $\bigcup_{c=_1}^\infty PCACHE_{n^c,B}$ is contained in $PSPACE$.  
\end{proof}

\subsection{$\CACHE_{M,B}$ Hierarchy}
\label{sec:Hierarchy}
In this section we prove that a hierarchy exists in the IO-model. The separation in the CACHE hierarchy is $B$ times the separation for the RAM hierarchy. We know that RAM machines given polynomially more time can solve more problems than those given polynomially less.

\begin{theorem}
	For $\epsilon \geq 0$,
	\begin{center}
		$RAM^O(t(n)) \subsetneq RAM^O((t(n))^{1+ \epsilon})$.
		\cite{cook1972time}
	\end{center}
	\label{thm:ramHierarchy}
\end{theorem}

Let $s(n)$ be the space usage of the algorithm running on the RAM machine. 
Let $\alpha = \frac{B+\lceil \lg\left(s\left(n\right)\right)/w\rceil}{B}$, which is the number of cache lines needed to represent both a cache line and its memory address. Note, in the case where one word is large enough to address all of the  memory used by the algorithm (a standard assumption) $\alpha = 1 +1/B \leq 2$. We now give a simulation of a CACHE machine by a RAM machine with MEM oracle.

%%%%%%%%%%%%%%%%%%%%%%%%%%%%%%%%%
\begin{lemma}
	$\CACHE_{M,B}(t\left(n\right)) \subseteq RAM \left (t\left(n\right)\left(B+\alpha \right) \right)^{MEM\left(Mw\alpha\right)}$
	\label{lem:NormalRamSimCache}
\end{lemma}
\begin{proof}
At a high level we are going to be treating the $\MEM$ oracle as the cache, the RAM machine is simply going to be playing the part of moving information from the main memory into the cache. 

We reserve the first $B$ words , ``the input'', of the $\MEM$ tape to be the location to write in a cache line to the cache simulation. We reserve the next $\lceil \lg(s(n))/w \rceil$ words to specify where this cache line came from in main memory.

 We reserve the next $\lceil \lg(s(n))/w \rceil$ words, ``the request'', to specify which cache line the cache simulation is requesting from main memory at the end of each run.
 
 Finally the next $B$ words, ``the output'', specify the contents of the cache line being kicked out of memory and the following $\lceil \lg(s(n))/w \rceil$ words specify where this cache line came from in main memory. 

When the cache simulation is run it takes the input and writes it and the $\lceil \lg(s(n))/w \rceil$ words of pointer information into the part of its $Mw\alpha$ sized tape where the output was previously written. Then the $\MEM(M\alpha w)$ oracle can compute the language which simulates a the cache until its next cache miss. 

Note this means we only need to copy words into memory when a cache miss occurs. The process of fetching the requested cache line, writing it to the input and, writing the output to our main memory takes $O(B+ \lceil \lg(s(n))/w \rceil)$ time. Note this is $O(B\alpha)$ time per cache miss, for a total of $O(t(n)B\alpha)$ time.
\end{proof}

Plugging our simulation into the RAM hierarchy gives a separation result for the CACHE complexity classes. 

\begin{theorem}
	For all $\epsilon >0 $
	$$\CACHE_{M,B}(t\left(n\right))  \subsetneq CACHE_{M\alpha,B}(\left(\alpha t\left(n\right)B \right)^{1+\epsilon}).$$
\end{theorem}
\begin{proof}
From Lemma \ref{lem:NormalRamSimCache}, 
$$\CACHE_{M,B}(t\left(n\right)) \subseteq RAM \left (t\left(n\right)\alpha B \right)^{MEM\left(Mw\alpha\right)}.$$

From Theorem \ref{thm:ramHierarchy}, for $\epsilon >0$
$$\CACHE_{M,B}(t\left(n\right)) \subsetneq RAM \left ((t\left(n\right)\alpha B)^{1+\epsilon} \right)^{MEM\left(Mw\alpha\right)}.$$

Using Corollary~\ref{cor:ram-subset-cache}:
$$\CACHE_{M,B}(t\left(n\right)) \subsetneq CACHE_{M\alpha, B} \left ((t\left(n\right)\alpha B)^{1+\epsilon} \right).$$

\end{proof}

%Can we get a sepearation with just time, memory, or cache line size. $M$ seems the most doable, consider looking at DTIME^SPACE(n) vs DTIME^SPACE(n^(1+\epsilon)).

Under reasonable assumptions about the values of input and word sizes, we can construct a cleaner version of the above theorem.

\begin{corollary}
When $s(n) = 2^{O(wB)}$, in other words the memory used by the algorithm is referenceable by $O(B)$ words, 
	$$\CACHE_{M,B}(t\left(n\right))  \subsetneq CACHE_{M,B}(\left( t\left(n\right)B \right)^{1+\epsilon}).$$

\end{corollary}
\begin{proof}
Note  $\alpha = 1 + 1/B = O(1)$, and thus is a constant with respect to the time and size of memory which are defined asymptotically. Thus this factor disappears.
\end{proof}

\subsection{TM Simulations for the RAM Model}
\label{sec:TMsimRAM}

Exploiting the cache line in the I/O-model is a long standing goal for many algorithms. Turing Machines have great locality and perform universal computation. Simulating Turing Machines in the I/O-model has the potential to give a universal transform which utilizes the cache line for improved speeds. Notably, improved simulations of RAM machines by multi-tape or multi-dimensional Turing Machines would imply savings factors of the cache line in running time.  

\subsection{Important simulations and separations}

Here we give some known results about simulations of RAM, $d$-dimensional Turing, and $c$-tape Turing Machines by each other. First we define a RAM machine with oracle access.

\begin{definition}
	$RAM^O\left(t\left(n\right)\right)$ is a RAM machine with oracle access to the language $O$ and allowed $t\left(n\right)$ time steps to do its computation. There is a separate location in memory where we can write down input to the oracle and receive output from the oracle.
\end{definition}

	%By the time hierarchy $DTIME \left( o \left ( \frac{f\left(n\right)}{\lg\left(f\left(n\right)\right)} \right) \right) \subsetneq DTIME \left(f\left(n\right) \right)$ \cite{hartmanis1965computational}. 
	
	We first give the known relativized time hierarchy result for Turing Machines which will provide the basis for a cache hierarchy of a different form. The time hierarchy proof relativizing means that the relationship between the classes remains the same with the introduction of an oracle, $O$. 
\begin{theorem}
	$DTIME^O \left( o \left ( \frac{f\left(n\right)}{\lg\left(f\left(n\right)\right)} \right) \right) \subsetneq DTIME^O \left(f\left(n\right) \right)$  [Fokelore].
\end{theorem}

A RAM machine can be simulated by a Turing Machine with a quadratic slowdown and consideration for word sizes. This simulation also holds true with respect to any oracle, $O$.
\begin{theorem}
	$RAM^O \left(f\left(n\right) \right) \subseteq DTIME^O\left(f\left(n\right)^{2}w/ \lg(f(n))\right)$ \cite{pippenger1979relations}.
	\label{thm:dtapesimofram}
\end{theorem}

Here we give a simulation of a RAM machine by a $d$-dimensional Turing Machine which also holds with respect to oracle access. The larger the dimension of the tape, the more efficient the simulation.
\begin{theorem}
	Let $DTIME_d(t(n))$ be the set of problems solvable with a d-dimensional Turing machine.  
	$RAM^O \left(f\left(n\right) \right) \subseteq DTIME_d^O\left(f\left(n\right)^{1+1/d}w/ \lg(f(n)) \right)$ \cite{pippenger1979relations}.
	\label{thm:ddimensionimofram}
\end{theorem}

\subsection{c-tape Turing Machine Simulations }

\begin{lemma}
	Let $DTIME_{(c)}$ be a multi-tape turing machine with $c$ tapes. Then 
	$$DTIME_{(c)}(t(n))^{MEM(Mw)} \subseteq CACHE_{M+2cB,B}(t(n)/(Bw)).$$
	\label{lem:cacheTMsimSaveBw}
\end{lemma}
\begin{proof}
	
	On the $c$ normal tapes we can maintain $2$ cache lines from each tape in cache. We start by keeping the $Bw$ bits before each tape head and the $Bw$ bits after each tape head in cache. If the tape head moves outside of this space, we keep the cache line closest to the head, kick out the cache line farther away, and finally bring in the $Bw$ bits containing the tape head and the closest $Bw$ bits currently uncovered (once again surrounding the head). Note that brining in $Bw$ bits takes one cache miss. We can lay out each tape in contiguous memory. 
	
We keep the entire $O(Mw)$ sized simulation of the $MEM(Mw)$ Oracle in memory. Now we can simulate the Turing machine with no cache misses, until a tape head on one of the $c$ tapes moves outside the area we are covering. For each tape head, the number of Turing Machine steps needed to cause a cache miss is at least $Bw$, in order to have the time to drag the tape head across the $Bw$ bits of tape. 
\end{proof}

%\begin{theorem}
%	\label{thm:ram-cache-b}
%	$$RAM(t(n))^{MEM(Mw)} \subseteq CACHE_{M,B}(t(n)^2/B)$$
%\end{theorem}
%\begin{proof}
%	By Theorem \ref{thm:dtapesimofram} we have that $RAM^{MEM(Mw)} \left(f\left(n\right) \right) \subseteq DTIME^{MEM(Mw)}\left(f\left(n\right)^{2}w\right)$.
%	Furthermore, by Lemma \ref{lem:cacheTMsimSaveBw} we have that $DTIME_{(c)}(t(n))^{MEM(Mw)} \subseteq CACHE_{M,B}(t(n)/(Bw))$. 
%	Combining these two results we get the desired statement.
%\end{proof}

If a c-tape TM can simulate a RAM machine very efficiently then we can save factors of $B$ (by getting memory locality). 
\begin{corollary}
	If $RAM(t(n))^O  \subseteq DTIME_{c-tape}(f_c(t(n)))^O$\\ for all oracles $O$
	then \\
	$RAM(t(n)) \subseteq CACHE_{M+2cB,B}(f_c(t(n))/B).$
\end{corollary}
\begin{proof}
Combine the assumption  with Lemma \ref{lem:cacheTMsimSaveBw} to get 
$$RAM(t(n))  \subseteq DTIME_{c-tape}(f_c(t(n))) \subseteq CACHE_{M+2Bc,B}\left( \frac{f_c(t(n))}{B}\right).$$
\end{proof}

This also has implications for the hierarchy theorem. 

\begin{lemma}
	If $RAM(t(n))^O  \subseteq DTIME_{c-tape}(f_c(t(n)))^O$\\ for all oracles $O$ 	
	then
	$$CACHE_{M,B}(t\left(n\right)) \subsetneq 
	CACHE_{\alpha M+2cB,B}\left (
	f_c\left(t\left(n\right)B\alpha\right) 
	\lg^2 \left(\frac{ f_c\left(t\left(n\right)B\alpha\right)}{B} \right)
	\right).$$
\end{lemma}
\begin{proof}
	Let $\alpha = \frac{B+\lceil \lg^2\left(s\left(n\right)\right)/w\rceil}{B}$.
We have by Lemma \ref{lem:NormalRamSimCache} that
$$CACHE_{M,B}(t\left(n\right)) \subseteq RAM \left (t\left(n\right)B\alpha \right)^{MEM\left(Mw\alpha\right)}.$$\\
By the assumed containment we have that
$$CACHE_{M,B}(t\left(n\right)) \subseteq DTIME_{c-tape}\left (f_c\left(t\left(n\right)B\alpha\right) \right)^{MEM\left(Mw\alpha\right)}.$$
By the time hierarchy 
$$CACHE_{M,B}(t\left(n\right)) \subsetneq 
DTIME_{c-tape}\left (
f_c\left(t\left(n\right)B\alpha\right) 
\lg^2 \left( f_c\left(t\left(n\right)B\alpha\right)  \right)
\right)^{MEM\left(Mw\alpha\right)}.$$
Thus,
$$CACHE_{M,B}(t\left(n\right)) \subsetneq 
CACHE_{\alpha M+2dB,B}\left (
f_c\left(t\left(n\right)B\alpha\right) 
\lg^2 \left( \frac{f_c\left(t\left(n\right)B\alpha\right)}{B} \right)
\right).$$
\end{proof}

\subsection{d-dimensional Turing Machine Simulations }

%In this section we provide another simulation of RAM computations by the CACHE model which saves a factor of $B$, the cache line size, in exchange for a factor of $t(n)$ the running time of the simulated machine. This is done by going through a Turing Machine simulation of the RAM computation which effectively localizes the computation allowing the cache to pull lines into memory in an efficient manner. This is a poor tradeoff in all but the most contrived cases, however we hope the technique might lead to a stronger result.

Next we show that strongly improved simulations of d-dimensional Turing Machines would imply saving polynomial factors of $B$. 

%\xxx{At least linear t(n) required. Put condition here.  }

%The next theorem gives a way of trading off simulated running time for some level of locality in a more fine grained manner. In practice one would want to determine whether $t(n)$ or $B$ is larger and either use Theorem~\ref{cor:ram-subset-cache} or Theorem~\ref{thm:ram-cache-b}, however we still think the ability to have this moderately smooth trade-off is interesting.

\begin{lemma}
	Let $DTIME_d$ be a d-dimensional Turing Machine and let $t(n) = \Omega(n)$. Then 
	$$DTIME_d(t(n))^{MEM(Mw)} \subseteq CACHE_{M,B}\left(\frac{t(n)}{\left(Bw\right)^{1/d}}\right).$$
	\label{lem:ddimensionSaveBw}
\end{lemma}
\begin{proof}
	This proof is similar to that in Lemma \ref{lem:cacheTMsimSaveBw}. We are going to break up the d-dimensional space into hypercubes of volume $Bw$ with edge length $(Bw)^{1/d}$. Whenever the TM transfers into a new region, we pull the surrounding $2^d$ blocks into memory (at least one of which must already be in memory). 
	
	This requires a cache of size $M + 2^dB$. Furthermore, it requires time $(2^d-1)t(n)/(Bw)^{1/d}$ to pull in these blocks. Saving a factor of $Bw/2^{d^2}$.
	
	So we have that 
	$$DTIME_d\left(t(n)\right)^{MEM(Mw)} \subseteq CACHE_{2^dM,B}\left(\frac{2^d t(n)}{\left(Bw\right)^{1/d}}\right).$$
	
	Given that $d$ is a constant the $2^d$ factors disappear as constant factors. 
	
\end{proof}
%%%%%%%%%%%%%%%%

\begin{lemma}
	
	If $RAM(t(n))^O  \subseteq DTIME_{d}(f_d(t(n)))^O$\\ for all oracles $O$, 	
	then

	$$RAM(t(n))^{MEM(Mw)} \subseteq CACHE_{M,B}\left(\frac{f_d(t(n))}{\left(wB\right)^{1/d}}\right)$$
	\label{lem:RamMemSimCacheMB}
\end{lemma}
\begin{proof}
	
	Our assumption gives us that
	$$RAM^{MEM(Mw)} \left(t\left(n\right) \right) \subseteq DTIME_d^{MEM(Mw)}\left(f_d\left(t\left(n\right)\right)\right).$$
	
	Then Lemma \ref{lem:ddimensionSaveBw} gives us that 
	$$RAM^{MEM(Mw)} \left(t\left(n\right) \right) \subseteq DTIME_d^{MEM(Mw)}\left(\frac{f_d\left(t\left(n\right)\right)}{ (wB)^{1/d}}\right).$$
\end{proof}

\begin{lemma}
	
	If $RAM(t(n))^O  \subseteq DTIME_{d}\left(f_d\left(t(n)\right)\right)^O$\\ for all oracles $O$ 	
	then, for all $\epsilon > 0$,

	$$CACHE_{M,B}\left(t(n)\right) \subsetneq 
	CACHE_{M,B}\left(
	\frac{f_d\left((\alpha t(n)B)^{1+\epsilon}\right)}
	{(wB)^{1/d}}
	\right)$$
	\label{lem:dtapeHyerarchy}
\end{lemma}
\begin{proof}
		Let $\alpha = \frac{B+\lceil \lg\left(s\left(n\right)\right)/w\rceil}{B}$.
		We have by Lemma \ref{lem:NormalRamSimCache} that
		$$CACHE_{M,B}\left(t\left(n\right)\right) \subseteq RAM \left( t\left(n\right)B\alpha \right)^{MEM\left(Mw\alpha\right)}.$$\\
		
		By the RAM time hierarchy (Theorem \ref{thm:ramHierarchy}) we have that, for all $\epsilon>0$, 
	$$CACHE_{M,B}(t\left(n\right))
	\subsetneq
	RAM \left( \left(t\left(n\right)B\alpha\right)^{1+\epsilon} \right)^{MEM\left(Mw\alpha\right)}.$$

	Then Lemma \ref{lem:RamMemSimCacheMB} gives us that 
	$$CACHE_{M,B}\left(t\left(n\right)\right) \subsetneq DTIME_d^{MEM(Mw\alpha)}\left(f_d((t\left(n\right)B\alpha)^{1+\epsilon})\right).$$
	
	Finally by Lemma \ref{lem:ddimensionSaveBw} we get that
		$$CACHE_{M,B}\left(t\left(n\right)\right) \subsetneq CACHE_{M\alpha, B}\left(f_d((t\left(n\right)B\alpha)^{1+\epsilon})/ (wB)^{1/d}\right).$$
\end{proof}

% \begin{theorem}
%	$$RAM(t(n))^{MEM(Mw)} \subseteq CACHE_{M,B}\left(\frac{(t(n))^{1+1/d}w^{1-1/d}}{B^{1/d}}\right)$$
%	\label{thm:RamMemSimCacheMB}
%\end{theorem}
%\begin{proof}
%	\xxx{TOOD: look up the d dimension simulation and check the w}	
%	
%	Theorem  \ref{thm:ddimensionimofram} gives us that
%	$RAM^{MEM(Mw)} \left(t\left(n\right) \right) \subseteq DTIME_d^{MEM(Mw)}\left(t\left(n\right)^{1+1/d}w\right)$.
%	
%	Then Lemma \ref{lem:ddimensionSaveBw} gives us that 
%	$$RAM^{MEM(Mw)} \left(t\left(n\right) \right) \subseteq CACHE_{M,B}( t\left(n\right)^{1+1/d}w^{1-1/d}/B^{1/d}).$$
%\end{proof} $

\subsection{Implications of Improved Simulations}
Note that with the current simulation efficiencies of c-tape Turing Machines (Theorem \ref{thm:dtapesimofram}) and d-dimensional Turing Machines (Theorem \ref{thm:ddimensionimofram}) we do not save factors of $B$ off our running time. Additionally, these simulations efficiencies do not tighten our hierarchy. 
However, some improvements to simulation efficiencies would improve our ability to automatically get locality and tighten our hierarchy. 

We now given an example of such an improvement assuming an improved simulation of  d-dimensional Turing Machines. Assume that for some constant $c$ we have that $RAM(t(n))^O  \subseteq DTIME_{d}(t(n)^{1+1/(cd)})^O$ for all oracles $O$.\\
 In what cases do we get an improved time? 
Let the running time of our RAM machine be $t(n) = n^k$ and the size of $B$ be $B = n^b$.
We save polynomial factors when simulating the RAM machine with a CACHE machine where $c > k/b$. To be more specific, if $c=8$ then we save polynomial factors when $B = n^{1/4}$ and $t(n) = n^{2-\epsilon}$ for $\epsilon>0$.

\section{Conclusion}
\label{sec:Conclusion}

In this paper we give a formal definition for a complexity class based on the I/O model of computation and show its relationship to other complexity classes. This gives us a bridge between these well studied fields. Our hierarchy separation gives further justification for the study of fine-grained complexity in the I/O model, and although we are able to transfer over some results, there is ample work to be done on this topic. Further our simulations suggest results in pure complexity theory could have implications for faster algorithms in the I/O model. From here we propose several specific problems for future work. 

We give a hierarchy, however, unfortunately, the separation not only includes an increase in running time but also cache and line size. The increase in size of cache and cache line are only a constant sized blow up under normal assumptions (that the word size is large enough to index the memory). However, it would be very interesting to show these results without any increase. Furthermore, removing the factor of $B$ from the hierarchy separation would be exciting. It would also be interesting to show a separation hierarchy based on cache size alone.

Many fine grained reductions in the RAM model port directly to the caching model. However, this need not be the case. Finding reductions between problems in the caching model (especially non-trivial ones) would be very interesting. Finding cases where the RAM and I/O reductions are very different would be interesting. Additionally, reducing between problems in the I/O model may lead to algorithmic improvements. 

We show a connection between Turing Machine simulations of RAM machines and memory locality in Section~\ref{sec:TMsimRAM}. Showing, through any method, that a certain factor of the cache line size, $B$, can always be saved in the I/O model would be very interesting.  
%We show that improved simulations of RAM machines with multi-tape or multi-dimensional Turing machines would imply savings in memory locality . This would also tighten the cache hierarchy. Giving a connection between complexity theory results and improvements to caching behavior could potentially result in practical speed ups. 

\section*{Acknowledgements}

We thank the anonymous reviewers for their helpful suggestions.

% Decrease the space between bibliography items.
\let\realbibitem=\bibitem
\def\bibitem{\par \vspace{-1.2ex}\realbibitem}

\bibliographystyle{alpha}
\bibliography{mybib} 

\newpage

\appendix

\section{New Upper Bounds Proofs}
\label{sec:newUpperProofs}

We will create a self reduction. The only problem we have is how to make sure that full adjacency lists are grouped together, unless they are too big for cache and then split. We will do this with division rules, and then argue that the splits aren't too inefficient.

We begin by building a more general algorithm that counts the nodes at distance $0$, $1$ and $2$. This will allow for computing 2vs3 Diameter and 2vs3 Radius efficiently (in Corollary \ref{cor:Diam} and Corollary \ref{cor:Radius}). 

\begin{theorem}
	In $O(|E|^2/(MB)+E\lg(E)/B)$ time we can return an array $T$ where for  $i\in[1,n]$ and $j\in[0,2]$ $T[i][j]$ is the number of nodes at distance $\leq j$ from node $i$ in an undirected, unweighted graph $G$. 
	\label{thm:diamCache}
	\label{thm:radiusCache}
\end{theorem}
\begin{proof}
	Start by setting $\forall i\in V$ $T[i][0]=1$. This takes time $O(|V|/B)$.

	Sort the adjacency lists of all nodes in time $sort(E) = O(E\lg(E)/B)$.
	Count the size of the adjacency lists of all nodes and mark these sizes in $O(E/B)$. 
	
	For each node set $T[i][1] = $ the length of its adjacency list plus $1$. This takes time $O(|V|/B)$.
	
	We will now solve for $T[i][2]$.
	
	Add to each node's adjacency list the node itself.
	Sort the nodes by the length of their adjacency lists in time $O(sort(E))$. 
	Sort the list $T[i][j]$ by vertex first and use this new order. This takes time $sort(V)$
	
	Start each list with the node name and a word of space to store the counter $distTwo$. Initialize $distTwo=0$. Append the sorted adjacency lists of each node in sorted order one after each other. Call this list of adjacency lists $A$. 
	Give each node an extra indicator bits $alreadyClose$. 
	
	Then we run the algorithm we will call $D$. We feed it two copies of the adjacency list $A$ to start, but we can feed two different lists. $D$ returns a tuple. The first value allows lower subproblems to propagate up that they found a large diameter. The second value is for message passing between levels of the program. For notational simplicity let $D(A,B)[0]$ be the diameter return value and $D(A,B)[1]$ be the message passing bit. 
	
	 Let $k_A=\begin{cases}
	|A|/2, & \text{if } |A|>|B|/2 \\
	|A|, & \text{if }|A|<|B|/2
	\end{cases}$ and $k_B=\begin{cases}
	|B|/2, & \text{if } |B|>|A|/2 \\
	|B|, & \text{if }|B|<|A|/2
	\end{cases}$. 
	
	Let $k'_A$ be the closest index into $A$ where two adjacency lists are split. Let $k'_B$ be the closest index into $B$ where two adjacency lists are split.
	
	So we will subdivide the problems, but some of our divisions will split the groups unevenly. This will limit how uneven our division can be. $D(A,B)$ has four cases.
	\begin{enumerate}
		\item If $A$ and $B$ are both subsets or all of one adjacency list then set the output bit $D(A,B)[1] =\\ D(A[0,k_A],B[0,k_B])[1]  \vee D(A[0,k_A],B[k_B+1,|B|])[1]  \vee D(A[k_A+1,|A|],B[0,k_B])[1]  \vee  D(A[k_A+1,|A|],B[k_B+1,|B|])[1] $. Base case is $|A|$ or $|B|$ is length 1. Then we simply scan through the other list for that one value. If they share a value return $True$ if not return $False$. If $A$ and $B$ are each a full adjacency list then add one to the counter $distTwo$ in $A$ and add one to the counter $distTwo$ in $B$.
		\item If $A$ is a subset of, or all of, one adjacency list and $B$ contains many or one adjacency lists then we will use the indicator bits on list $B$. Intuitively, we are going to set these indicator bits to say if the subset of $A$ we are looking at has overlapped with $B$.
		
		We call $D(A[0,k_A],B[0,k_B])$,  $D(A[0,k_A],B[k_B'+1,|B|])$, $D(A[k_A+1,|A|],B[0,k_B'])$ and finally $D(A[k_A+1,|A|],B[k_B'+1,|B|])$. The base case is $B$ contains only one adjacency list. In this case we set $alreadyClose$ of this list to $D(A[0,k_A],B[0,k_B])[1]  \vee D(A[0,k_A],B[k_B+1,|B|])[1]  \vee D(A[k_A+1,|A|],B[0,k_B])[1]  \vee  D(A[k_A+1,|A|],B[k_B+1,|B|])[1] $, where all these will be sublists (and thus call function $1$). 
		
		If $A$ is the entirety of one list then scan $B$ for all the nodes $i\in B$ the $i.alreadyClose$ bits. Let $a.distTwo$ be the $distTwo$ variable from $A$. For each node $i$ where $i.alreadyClose=True$ add one to the $i.distTwo$ and add one to $a.distTwo$. 
		\item If $A$ contains one or many full adjacency lists and $B$ contains a subset or all of one adjacency list, then call $D(B,A)$.
		\item  If $A$ and $B$ both contain one or more full adjacency lists then we  set $k_A'$ and $k_B'$. Then we return
		$D(A[0,k_A'],B[0,k_B'])[1]  \wedge D(A[0,k_A'],B[k_B'+1,|B|])[1]  \wedge D(A[k_A'+1,|A|],B[0,k_B'])[1]  \wedge  D(A[k_A'+1,|A|],B[k_B'+1,|B|])[1] $
	\end{enumerate}
	
	We scan through the adjacency list and set $T[i][2]=i.distTwo$ for all $i$. We add counts when a list is being considered and is about to be divided. So we never double count. We iterate to the bottom and thus every list will eventually be the entirety of the input $A$ at some level of the recursion. This takes time $O(|E|/B)$ because the lists are in the same order due to our previous sort. We now have all the counts. 
	
	Once a subproblem fits in memory (that is $|A|+|B|<M$), whatever size that is, the subproblem is solved in the time to read in and write out all the data, so $O(M/B)$.  
	
	We can sub-divide unevenly when we are splitting up many adjacency lists, because we choose to divide adjacency lists where they split. However, we can count the total number of subproblems that fit in memory. For adjacency lists, $L$, of length $\geq M/4$ we will get $\leq 2L/M+1$ sublists. For adjacency lists, $L$, of length $ \leq M/4$ at most one will be alone. The other lists will be in sublists with $\geq M/(2L)$ other sublists. Thus, we will have at most twice as many sublists as we ought to, thus we will have at most four times as many subproblems of size $\leq M/2$ as we should. 
	
	At every instance we divide into at least two new subproblems. Above we bounded the number of subproblems produced at $O(|E|^2/M)$. Thus the maximum number of layers in our self recurrence is $O(\lg(E))$. We take scan time per layer. The total time for these scans is thus $O(|E|\lg(|E|)/B)$. So the total time for this algorithm will be $O((|E|/M)^2M/B+|E|\lg(|E|)/B+sort(|E|)) = O(|E|^2/(MB)+|E|\lg(|E|)/B)$. 
\end{proof}

\begin{corollary}
In $O(|E|^2/(MB)+E\lg(E)/B)$ time we can return an array $S$ where for  $i\in[1,n]$ and $j\in[0,2]$ $S[i][j]$ is the number of nodes at distance $j$ from node $i$ in an undirected, unweighted graph $G$. Also add a column for all $i\in[1,n]$  $S[i][\geq3]$ which lists the number of nodes at distance greater than or equal to 3. 
\label{cor:Sdiamradius} 
\end{corollary}
\begin{proof}
	First use Theorem \ref{thm:diamCache} to get $T$. Then we scan through $T$ and generate $S$.
	\begin{align}
	S[i][0]&= T[i][0]\\
	S[i][1]&= T[i][1]-T[i][0]\\
	S[i][2]&= T[i][2]-T[i][1]+T[i][0]\\
	S[i][3]&= n-T[i][2]\\
	\end{align}
	
	This takes time $O(|V|/B + |E|^2/(MB)+E\lg(E)/B)$ which is $O(|E|^2/(MB)+E\lg(E)/B)$.
\end{proof}

\paragraph{Diameter 2 vs 3}
\begin{corollary}
	In $O(|E|^2/(MB)+E\lg(E)/B)$ time we can determine if the diameter of an undirected unweighted graph is $1$, $2$, or $\geq 3$. \label{cor:Diam}
\end{corollary}
\begin{proof}
	First use Corollary \ref{cor:Sdiamradius} to get $S$. Then we scan through $S$. If $S[i][\geq 3]>0$ for any $i$ then the Diameter is $\geq 3$, return $\geq 3$. Else if $S[i][2]>0$ for any $i$ then the Diameter is $2$, return $2$. Else if $S[i][1]>0$ for any $i$ then the Diameter is $1$, return $1$. 
	
	This takes time $O(|V|/B + |E|^2/(MB)+E\lg(E)/B)$ which is $O(|E|^2/(MB)+E\lg(E)/B)$.
\end{proof}

\paragraph{2 vs 3 Radius}
\begin{corollary}
	In $O(|E|^2/(MB)+E\lg(E)/B)$ time we can determine if the radius of an undirected unweighted graph is $1$, $2$, or $\geq 3$. \label{cor:Radius}
\end{corollary}
\begin{proof}
	First use Theorem \ref{thm:diamCache} to get $T$. Then we scan through $S$. If $T[i][1]=n$ for any $i$ then the radius is $1$, return $1$. Else if $T[i][2]=n$ for any $i$ then the Diameter is $2$, return $2$. Else return $\geq 3$. 
	
	This takes time $O(|V|/B + |E|^2/(MB)+E\lg(E)/B)$ which is $O(|E|^2/(MB)+E\lg(E)/B)$.
\end{proof}

\end{document}